\documentclass[11pt]{extarticle}
\usepackage{amsmath,amssymb,amsfonts,amsthm}
\usepackage{mathtools}
\usepackage{mathrsfs}
\usepackage{enumitem}
\usepackage[margin=1in]{geometry}
\usepackage{hyperref}

\newtheorem{theorem}{Theorem}[section]
\newtheorem{lemma}[theorem]{Lemma}
\newtheorem{proposition}[theorem]{Proposition}
\newtheorem{corollary}[theorem]{Corollary}
\theoremstyle{definition}
\newtheorem{definition}[theorem]{Definition}
\theoremstyle{remark}
\newtheorem{remark}[theorem]{Remark}

\numberwithin{equation}{section}

\newcommand{\F}{\mathbb F}
\newcommand{\N}{\mathbb N}
\newcommand{\Z}{\mathbb Z}
\newcommand{\Oq}{O_q}
\newcommand{\Aq}{\mathcal A_q}
\newcommand{\Aqfib}[1]{\mathcal A_{q,#1}}
\newcommand{\BBqquot}[2]{\widetilde{\mathcal A}_{q,#1}^{\{#2\}}}
\newcommand{\Uplus}{U_q^+(\widehat{\mathfrak{sl}}_2)}
\newcommand{\Vsh}{\mathbb V}
\newcommand{\gr}{\operatorname{gr}}
\newcommand{\init}{\operatorname{in}}
\newcommand{\sh}{\mathrm{sh}}
\newcommand{\cW}{\mathcal W}
\newcommand{\cG}{\mathcal G}
\newcommand{\ctG}{\widetilde{\mathcal G}}
\newcommand{\Zvee}{\mathcal Z^{\vee}}
\newcommand{\Zconv}{\mathcal Z}
\newcommand{\qcomm}[2]{[#1,#2]_q}
\newcommand{\comm}[2]{[#1,#2]}
\DeclareMathOperator{\Cent}{Cent}

\begin{document}

\title{PBW bases and centralisers for the $q$-Onsager algebra}

\author{Haoran Zhu\thanks{Division of Mathematical Sciences,
Nanyang Technological University, 21 Nanyang Link, Singapore 637371.
Email: \texttt{zhuh0031@e.ntu.edu.sg}}}

\date{}

\maketitle

\begin{abstract}
We prove that, over an arbitrary field and whenever $q$ is not a root of
unity, the Baseilhac--Kolb root vectors form a PBW basis of the
$q$-Onsager algebra for every total order on the positive roots of
$\widehat{\mathfrak{sl}}_2$.  This removes the previous transcendence hypothesis.
We establish twelve PBW bases in the alternating generators and show that
they persist under arbitrary scalar central specialisation of the
alternating central extension.  We determine the centraliser of the
negative alternating subalgebra and that of the first imaginary alternating
generator, and deduce that the four single-family alternating polynomial
subalgebras are maximal commutative.  Together, these results settle four
conjectures of Terwilliger and, in characteristic different from $2$, a
conjecture of Baseilhac and Belliard.
\end{abstract}


\medskip


\section{Introduction}
\label{sec:introduction}

The Onsager algebra first appeared in Onsager's solution of the
two-dimensional Ising model \cite{Onsager1944}.  Its two-generator
presentation by the Dolan--Grady relations provides an algebraic
construction of the commuting quantities that occur in Onsager-type
integrability \cite{DolanGrady1982}.  The $q$-Onsager algebra $\Oq$ is the
corresponding $q$-deformation.  It appears in boundary quantum integrable
systems
\cite{Sklyanin1988,Baseilhac2005,BaseilhacKoizumi2005,
BaseilhacBelliard2010,BaseilhacBelliard2013},
in the theory of tridiagonal algebras \cite{Terwilliger2001}, and as a
basic example of a quantum symmetric-pair coideal subalgebra
\cite{Letzter2002,Kolb2014,LuWang2021,LuRuanWang2023}.

There are two principal PBW constructions for $\Oq$.  Baseilhac and Kolb
used braid-group automorphisms to define real and imaginary root vectors,
modelled on Damiani's root vectors for
$U_q^+(\widehat{\mathfrak{sl}}_2)$, and obtained a root-vector PBW basis
\cite{Damiani1993,BaseilhacKolb2020}.  A second construction comes from the
current algebra associated with the reflection equation
\cite{BaseilhacKoizumi2005,BaseilhacShigechi2010,
BaseilhacBelliardCrampe2018}.  Terwilliger identified this current algebra
with the alternating central extension $\Aq$ of $\Oq$ and proved the
factorisation $\Aq\simeq\Oq\otimes\F[z_1,z_2,\ldots]$
\cite{Terwilliger2021ACE,Terwilliger2022OqACE,
Terwilliger2023Compact}.  The four current families descend under the
canonical central reduction to the alternating generators of $\Oq$.  We
write
\begin{equation*}
\begin{aligned}
W^-&=\{W_{-n}\mid n\geqslant 0\},&
W^+&=\{W_{n+1}\mid n\geqslant 0\},\\
G&=\{G_{n+1}\mid n\geqslant 0\},&
\widetilde G&=\{\widetilde G_{n+1}\mid n\geqslant 0\}.
\end{aligned}
\end{equation*}

The reflection-equation presentation of $\Aq$ has central coefficients
$\Delta_1,\Delta_2,\ldots$ arising from the quantum determinant.
Baseilhac and Belliard fixed their scalar values and considered the
resulting quotients of the current algebra.  They proved that every current
generator in such a quotient can be recovered recursively from $W_0,W_1$
\cite[Proposition~3.1 and Corollary~3.1]{BaseilhacBelliard2017}.  They then
conjectured that the ordered monomials in the three blocks
$W^-<G<W^+$ form a PBW basis
\cite[Conjecture~1]{BaseilhacBelliard2017}; they called this the
\textbf{$WG$-basis}.  The conjecture gives more than generation by
$W_0,W_1$: it gives a unique normal form in the current coefficients after
the central quantum determinant has been assigned scalar values.
Baseilhac and Belliard supported the conjecture by a Hilbert-series
comparison and low-degree transition matrices
\cite[Section~4 and Appendix~B]{BaseilhacBelliard2017}.

Terwilliger's factorisation of $\Aq$ identifies every scalar central fibre
with $\Oq$, but it does not by itself prove the $WG$-basis conjecture.
A central specialisation changes each alternating generator by lower-index
terms, and a four-family PBW basis of $\Aq$ does not formally give a
three-family PBW basis in a quotient.  Terwilliger therefore proposed six
three-family orders for $W^-,\widetilde G,W^+$
\cite[Conjecture~16.2]{Terwilliger2022OqACE}.  He also asked whether the
negative and imaginary alternating subalgebras are self-centralising
\cite[Conjectures~16.7 and~16.8]{Terwilliger2022OqACE}.  More recently, the
author proved the four-family PBW order
$\mathcal W^-<\mathcal G<\widetilde{\mathcal G}<\mathcal W^+$ for $\Aq$
\cite{Zhu2026}.  The present paper treats the corresponding three-family
bases after central specialisation, together with the root-vector and
centraliser problems for $\Oq$.

The order in a PBW basis is part of the structure.  Factorised universal
$R$-matrices are associated with specified convex orders
\cite{KhoroshkinTolstoy1991,Beck1994,Damiani1998}, and the same order is
visible in Gauss or LDU decompositions of $L$-operators
\cite{DingFrenkel1993}.  In the coideal setting, universal and tensor
$K$-matrices provide the corresponding structures
\cite{BalagovicKolb2019,AppelVlaar2022,AppelVlaar2025Tensor}.
For the $q$-Onsager algebra, boundary and fused $K$-operators occur in
closely related constructions
\cite{BaseilhacTsuboi2018,LemartheBaseilhacGainutdinov2026}.
Alternating normal orders also enter recent calculations of universal
TT- and TQ-relations for $\Aq$
\cite{BaseilhacGainutdinovLemarthe2026}.  A PBW theorem in every scalar
central fibre therefore provides a normal form for these current
coefficients after the central parameters have been fixed.

Throughout the paper, $\F$ denotes a field and $0\neq q\in\F$ is not a
root of unity, unless stated otherwise.  Our first result concerns the
Baseilhac--Kolb root vectors.

\begin{theorem}[{\cite[Conjecture~16.1]{Terwilliger2022OqACE}}]
\label{thm:intro-root}
For every total order on the positive roots of
$\widehat{\mathfrak{sl}}_2$, the corresponding ordered monomials in the
Baseilhac--Kolb root vectors form an $\F$-basis of $\Oq$.
\end{theorem}

The proof has two steps.  We first use Damiani's straightening relations to
show that the Damiani root vectors give a PBW basis in every total order.
For each fixed weight, every correction term decreases a lexicographically
ordered real-root multiplicity sequence, so induction on this sequence and
on the number of inversions terminates.  We then compute the initial forms
of the Baseilhac--Kolb root vectors under the word-length filtration
$\gr\Oq\simeq\Uplus$ and apply a filtered lifting argument.  This both
permits an arbitrary total order and removes the previous transcendence
hypothesis.

We next consider the alternating generators.  For
$H\in\{G,\widetilde G\}$, consider the six block orders
\begin{equation}
\label{eq:intro-orders}
\begin{array}{lll}
W^-<H<W^+,&
W^+<H<W^-,&
W^+<W^-<H,\\[1mm]
W^-<W^+<H,&
H<W^+<W^-,&
H<W^-<W^+.
\end{array}
\end{equation}

\begin{theorem}[{\cite[Conjecture~16.2]{Terwilliger2022OqACE}}]
\label{thm:intro-alternating}
Let $H=G$ or $H=\widetilde G$.  Every total order on
$W^-\cup H\cup W^+$ satisfying one of the block conditions in
\eqref{eq:intro-orders} gives a PBW basis of $\Oq$.  The corresponding
twelve PBW bases persist for the specialised alternating generators in
every $\F$-valued central fibre of $\Aq$.
\end{theorem}

To prove Theorem~\ref{thm:intro-alternating}, we determine the initial
form of every alternating generator, including its scalar factor, and
identify it with the corresponding alternating word in the $q$-shuffle
realisation of $\Uplus$.  The twelve graded PBW bases are due to
Terwilliger \cite{Terwilliger2019Alternating}, and the result follows by
filtered lifting.  Under scalar central specialisation, each alternating
generator is changed only by terms of smaller filtration degree, so its
initial form is unchanged.

For $\rho\in\F^\times$ and
$\boldsymbol\delta=(\delta_1,\delta_2,\ldots)$, let
\begin{equation*}
 \BBqquot{\rho}{\boldsymbol\delta}
 =
 \Aq(\rho)\big/
 \left\langle
        \Delta_n^{(\rho)}-2\delta_n\mid n\geqslant 1
 \right\rangle
\end{equation*}
denote the quotient introduced by Baseilhac and Belliard.

\begin{corollary}[{\cite[Conjecture~1]{BaseilhacBelliard2017}}]
\label{cor:intro-BB}
Assume that $\operatorname{char}\F\neq 2$.  For every
$\rho\in\F^\times$ and every scalar sequence
$\boldsymbol\delta$, the ordered monomials in the images of the three
families $W^-<G<W^+$ form an $\F$-basis of
$\BBqquot{\rho}{\boldsymbol\delta}$.
\end{corollary}

For the normalised parameter, Corollary~\ref{cor:intro-BB} follows from
the central-fibre theorem after comparing the three systems of central
coordinates.  The characteristic restriction is needed only for the
comparison with the Baseilhac--Belliard coordinates $\Delta_n$; the
central-fibre theorem itself is valid in arbitrary characteristic.  The
general non-zero parameter is obtained by scalar extension, rescaling the
current generators, and descent by faithful flatness.

Our final result concerns centralisers.  Each of the four alternating
families is commutative, and we determine the centralisers needed to test
maximal commutativity.

\begin{theorem}[{\cite[Conjectures~16.7 and~16.8]{Terwilliger2022OqACE}}]
\label{thm:intro-centralisers}
One has
\begin{align*}
 \Cent_{\Oq}\bigl(\F[W_0,W_{-1},W_{-2},\ldots]\bigr)
 &=
 \F[W_0,W_{-1},W_{-2},\ldots],\\
 \Cent_{\Oq}(\widetilde G_1)
 &=
 \F[\widetilde G_1,\widetilde G_2,\widetilde G_3,\ldots].
\end{align*}
Consequently, the four polynomial subalgebras generated separately by the
four alternating families are self-centralising and hence maximal
commutative.
\end{theorem}

For the negative family, we work first in the $q$-shuffle algebra and
introduce a left-deletion operator.  This reduces the joint centraliser of
$W_0=x$ and $W_{-1}=xyx$ to the centraliser of $G_1=yx$.  An auxiliary
Damiani-basis argument and a filtration by the number of $G$-factors then
show that this joint centraliser is exactly the negative alternating
polynomial algebra.  Hence that algebra is self-centralising, and the
result lifts to $\Oq$ by the filtered centraliser lemma.  For the imaginary
family, we order the Baseilhac--Kolb real root vectors as a two-sided chain
and use the action of the first imaginary root vector.  Raising the largest
real index produces a uniquely largest profile with non-zero coefficient.

The paper is organised as follows.  In Section~\ref{sec:prelim}, we define
$\Oq$ and $\Aq$ and prove the filtered lifting lemmas used throughout.  In
Section~\ref{sec:qshuffle}, we recall the $q$-shuffle realisation, the
alternating words, the Catalan elements, and Damiani's root vectors.  In
Section~\ref{sec:specialisation}, we prove the arbitrary-order Damiani PBW
theorem.  In Section~\ref{sec:root-initial}, we compute the initial forms
of the Baseilhac--Kolb root vectors and prove
Theorem~\ref{thm:intro-root}.  In
Section~\ref{sec:alternating-initial}, we compute the initial forms of the
alternating generators, and in Section~\ref{sec:twelve}, we prove the
twelve PBW bases of $\Oq$.  In
Section~\ref{sec:central-specialisations}, we treat arbitrary scalar
central fibres, and in Section~\ref{sec:BB}, we prove
Corollary~\ref{cor:intro-BB}.  In Section~\ref{sec:negative-centraliser}, we
determine the negative alternating centraliser, and in
Section~\ref{sec:centralisers}, we determine the imaginary alternating
centraliser and prove maximal commutativity.

\section{The \texorpdfstring{$q$}{q}-Onsager algebra and filtered lifting}
\label{sec:prelim}

Throughout, we use the convention $\N=\{0,1,2,\ldots\}$.
The symbol $\F$ denotes a field, and
$q\in\F^\times$ is assumed not to be a root of unity.  All algebras and
algebra homomorphisms are unital, and all tensor products are taken over
$\F$.  For elements $X,Y$ of an $\F$-algebra, we write
\[
 [n]_q=\frac{q^n-q^{-n}}{q-q^{-1}},\qquad
 \comm{X}{Y}=XY-YX,\qquad
 \qcomm{X}{Y}=qXY-q^{-1}YX.
\]

We shall repeatedly use the following elementary consequence of the
assumption on $q$.

\begin{lemma}
\label{lem:nonvanishing-scalars}
For every positive integer $m$, the scalars
\[
 q-q^{-1},\qquad q+q^{-1},\qquad [m]_q,\qquad
 q^m+q^{-m},\qquad 1-q^{2m},\qquad 1-q^{-2m}
\]
are non-zero.
\end{lemma}

\begin{proof}
If $q-q^{-1}=0$, then $q^2=1$.  If
$q+q^{-1}=0$, then $q^4=1$.  Since $q-q^{-1}\neq 0$, the equality
$[m]_q=0$ implies $q^{2m}=1$.  Similarly,
$q^m+q^{-m}=0$ implies $q^{4m}=1$, while either
$1-q^{2m}=0$ or $1-q^{-2m}=0$ implies $q^{2m}=1$.  Each
possibility contradicts the assumption that $q$ is not a root of unity.
\end{proof}

\subsection{The \texorpdfstring{$q$}{q}-Onsager algebra and its alternating central extension}

\begin{definition}
\label{def:Oq}
The \textbf{$q$-Onsager algebra} $\Oq$ is the $\F$-algebra generated by
$W_0,W_1$, subject to the $q$-Dolan--Grady relations
\begin{align}
W_0^3W_1-[3]_qW_0^2W_1W_0+[3]_qW_0W_1W_0^2-W_1W_0^3
   &=(q^2-q^{-2})^2(W_1W_0-W_0W_1),
\label{eq:qDG0}\\
W_1^3W_0-[3]_qW_1^2W_0W_1+[3]_qW_1W_0W_1^2-W_0W_1^3
   &=(q^2-q^{-2})^2(W_0W_1-W_1W_0).
\label{eq:qDG1}
\end{align}
\end{definition}

Set
\begin{equation}
\label{eq:rho0}
        \rho_0=-(q^2-q^{-2})^2.
\end{equation}
Equivalently, the right-hand sides of
\eqref{eq:qDG0} and \eqref{eq:qDG1} are
$\rho_0\comm{W_0}{W_1}$ and
$\rho_0\comm{W_1}{W_0}$, respectively.

We shall also use the alternating central extension of $\Oq$, in the
presentation of
\cite{BaseilhacKoizumi2005,BaseilhacShigechi2010,
Terwilliger2021ACE}.

\begin{definition}
\label{def:Aq}
The \textbf{alternating central extension} $\Aq$ is the $\F$-algebra generated
by the four families
\[
        \cW_{-k},\qquad
        \cW_{k+1},\qquad
        \cG_{k+1},\qquad
        \ctG_{k+1}
        \qquad(k\in\N),
\]
subject, for all $k,\ell\in\N$, to
\begin{align}
&[\cW_0,\cW_{k+1}]
        =[\cW_{-k},\cW_1]
          =\frac{\ctG_{k+1}-\cG_{k+1}}{q+q^{-1}},
          \label{eq:Aq1}\\
&[\cW_0,\cG_{k+1}]_q
        =[\ctG_{k+1},\cW_0]_q
          =\rho_0\cW_{-k-1}-\rho_0\cW_{k+1},
          \label{eq:Aq2}\\
&[\cG_{k+1},\cW_1]_q
        =[\cW_1,\ctG_{k+1}]_q
          =\rho_0\cW_{k+2}-\rho_0\cW_{-k},
          \label{eq:Aq3}\\
&[\cW_{-k},\cW_{-\ell}]
        =0,
        \qquad
        [\cW_{k+1},\cW_{\ell+1}]=0,
        \label{eq:Aq4}\\
&[\cW_{-k},\cW_{\ell+1}]
        +[\cW_{k+1},\cW_{-\ell}]=0,
        \label{eq:Aq5}\\
&[\cW_{-k},\cG_{\ell+1}]
        +[\cG_{k+1},\cW_{-\ell}]=0,
        \label{eq:Aq6}\\
&[\cW_{-k},\ctG_{\ell+1}]
        +[\ctG_{k+1},\cW_{-\ell}]=0,
        \label{eq:Aq7}\\
&[\cW_{k+1},\cG_{\ell+1}]
        +[\cG_{k+1},\cW_{\ell+1}]=0,
        \label{eq:Aq8}\\
&[\cW_{k+1},\ctG_{\ell+1}]
        +[\ctG_{k+1},\cW_{\ell+1}]=0,
        \label{eq:Aq9}\\
&[\cG_{k+1},\cG_{\ell+1}]
        =0,
        \qquad
        [\ctG_{k+1},\ctG_{\ell+1}]=0,
        \label{eq:Aq10}\\
&[\ctG_{k+1},\cG_{\ell+1}]
        +[\cG_{k+1},\ctG_{\ell+1}]=0.
        \label{eq:Aq11}
\end{align}
The four displayed families are called the \textbf{alternating generators} of $\Aq$.
\end{definition}

We use the scalar convention
\begin{equation}
\label{eq:Aq-G0}
        \cG_0=\ctG_0=-(q-q^{-1})[2]_q^2.
\end{equation}
Both symbols in \eqref{eq:Aq-G0} denote this scalar; they are not
additional generators.

Terwilliger proved that $Z(\Aq)$ is a polynomial algebra in countably
many variables and that there is an algebra isomorphism
\begin{equation}
\label{eq:standard-factorisation-intro}
        \Aq\simeq\Oq\otimes\F[z_1,z_2,\ldots].
\end{equation}
See \cite{Terwilliger2021ACE,Terwilliger2022OqACE}.  Under this
factorisation, the canonical epimorphism
$\gamma:\Aq\longrightarrow\Oq$ corresponds to evaluation at
$z_1=z_2=\cdots=0$.  The alternating generators of $\Oq$ are the
images
\begin{equation*}
\begin{aligned}
 W_{-n}&=\gamma(\cW_{-n}),&
 W_{n+1}&=\gamma(\cW_{n+1}),\\
 G_{n+1}&=\gamma(\cG_{n+1}),&
 \widetilde G_{n+1}&=\gamma(\ctG_{n+1})
 \qquad(n\in\N).
\end{aligned}
\end{equation*}
We use calligraphic symbols for elements of $\Aq$, and the corresponding
plain symbols for their images in $\Oq$.

\subsection{The word-length filtration}

For $d\in\N$, let $F_d\Oq$ be the $\F$-subspace spanned by the
images of the words in $W_0,W_1$ of length at most $d$, and put
$F_{-1}\Oq=0$.  Then
$F_r\Oq\,F_s\Oq\subseteq F_{r+s}\Oq$ and
$\Oq=\bigcup_{d\geqslant 0}F_d\Oq$.  We call this the
\textbf{word-length filtration}.  Its associated graded algebra is
\[
        \gr\Oq
        =
        \bigoplus_{d\geqslant 0}F_d\Oq/F_{d-1}\Oq.
\]

For $u\in F_d\Oq$, write
$[u]_d=u+F_{d-1}\Oq\in F_d\Oq/F_{d-1}\Oq$.  When $d$ is clear,
we may write $\overline u$ in place of $[u]_d$.  For
$0\neq u\in\Oq$, put
\[
        \deg_F u=\min\{d\in\N\mid u\in F_d\Oq\}
\]
and call $\init(u)=[u]_{\deg_F u}$ the \textbf{initial form} of $u$.

Let $\Uplus$ be the $\F$-algebra generated by $A,B$, subject to the
cubic $q$-Serre relations
\begin{align*}
A^3B-[3]_qA^2BA+[3]_qABA^2-BA^3&=0,\\
B^3A-[3]_qB^2AB+[3]_qBAB^2-AB^3&=0.
\end{align*}
These are the highest homogeneous components of the
$q$-Dolan--Grady relations.  Terwilliger proved that $\Uplus$ is
precisely the associated graded algebra of $\Oq$
\cite[Theorem~4.4]{Terwilliger2017PositivePart}.

\begin{theorem}
\label{thm:gr-isomorphism}
There is an algebra isomorphism
\[
 \psi:\Uplus\longrightarrow\gr\Oq,
 \qquad
 \psi(A)=[W_0]_1,
 \qquad
 \psi(B)=[W_1]_1.
\]
\end{theorem}

From now on, we identify $\gr\Oq$ with $\Uplus$ through $\psi$.

\subsection{Filtered lifting}

We use the following convention for PBW bases.

\begin{definition}
Let $R$ be an $\F$-algebra, let $\Omega\subseteq R$, and let $<$
be a total order on $\Omega$.  We say that $\Omega$, with the order
$<$, gives \textbf{a PBW basis} of $R$ if the monomials
\[
        a_1a_2\cdots a_m,
        \qquad
        m\in\N,\qquad
        a_1\leqslant a_2\leqslant \cdots\leqslant a_m,
\]
form an $\F$-basis of $R$.  The monomial with $m=0$ is interpreted
as $1$.
\end{definition}

The following lifting lemma will be used repeatedly.

\begin{lemma}
\label{lem:filtered-lifting}
Let $R$ be an $\F$-algebra with an exhaustive algebra filtration
\[
        0=F_{-1}R\subseteq F_0R\subseteq F_1R\subseteq\cdots,
        \qquad
        F_rR\,F_sR\subseteq F_{r+s}R.
\]
Let $\{u_\lambda\}_{\lambda\in\Lambda}$ be elements of positive
filtration degree, say
\[
        u_\lambda\in
        F_{d_\lambda}R\setminus F_{d_\lambda-1}R,
        \qquad
        d_\lambda>0.
\]
Suppose that
\[
        \init(u_\lambda)=c_\lambda v_\lambda,
        \qquad
        c_\lambda\in\F^\times,
\]
where $v_\lambda\in\gr R$ is homogeneous of degree $d_\lambda$.
Fix a total order on $\Lambda$.  If the ordered monomials in the
$v_\lambda$ form a basis of $\gr R$, then the corresponding ordered
monomials in the $u_\lambda$ form a basis of $R$.
\end{lemma}

\begin{proof}
For an ordered monomial $M$, let $M(u)$ and $M(v)$ denote the
corresponding products in the $u_\lambda$ and $v_\lambda$,
respectively.  Write $e_\lambda(M)$ for the multiplicity of
$u_\lambda$ in $M(u)$, and put $d(M)=\sum_\lambda e_\lambda(M)d_\lambda$.
The image of $M(u)$ in
$F_{d(M)}R/F_{d(M)-1}R$ is
\[
        \left(
        \prod_\lambda c_\lambda^{e_\lambda(M)}
        \right)M(v).
\]
This element is non-zero, since $M(v)$ is a basis element of
$\gr R$.  Hence $M(u)$ has filtration degree $d(M)$, and
\begin{equation}
\label{eq:filtered-lifting-monomial}
 \init(M(u))
 =
 \left(
 \prod_\lambda c_\lambda^{e_\lambda(M)}
 \right)M(v).
\end{equation}

Suppose that $\sum_M a_MM(u)=0$ is a non-trivial finite relation among distinct ordered monomials.  Choose
$D$ maximal such that $a_M\neq 0$ for some $M$ with $d(M)=D$.
Taking the degree-$D$ component and using
\eqref{eq:filtered-lifting-monomial} gives
\[
 \sum_{d(M)=D}
 a_M
 \left(
 \prod_\lambda c_\lambda^{e_\lambda(M)}
 \right)M(v)=0.
\]
The monomials $M(v)$ are distinct basis elements and all displayed
scalars are non-zero.  Thus $a_M=0$ whenever $d(M)=D$, contradicting
the choice of $D$.  The monomials $M(u)$ are therefore linearly
independent.

It remains to prove that they span $R$.  We argue by induction on $D$.
Let $r\in F_DR$.  Its image in
$F_DR/F_{D-1}R$ is a finite linear combination of ordered monomials
$M(v)$ of degree $D$.  By
\eqref{eq:filtered-lifting-monomial}, we may subtract from $r$ a linear
combination of the corresponding monomials $M(u)$ so that the difference
belongs to $F_{D-1}R$.  The induction hypothesis completes the proof.
\end{proof}

\begin{remark}
Lemma~\ref{lem:filtered-lifting} is an initial-form version of
\cite[Proposition~4.6]{Terwilliger2017PositivePart}.  The formulation
above does not require the elements $u_\lambda$ to be specified in
advance by homogeneous words in a free algebra.
\end{remark}

For a subset $S$ of an algebra $R$, write
$\Cent_R(S)=\{x\in R\mid xs=sx\text{ for every }s\in S\}$.  We shall
also use the following centraliser version of filtered lifting.

\begin{lemma}
\label{lem:filtered-centraliser}
Let $R=\bigcup_{d\geqslant 0}F_dR$ be an exhaustively filtered algebra, with
$F_{-1}R=0$, and let $H\subseteq R$ be a commutative subalgebra with
the induced filtration $H\cap F_dR$.  Regard
\[
 \gr H
 =
 \bigoplus_{d\geqslant 0}
 \frac{H\cap F_dR}{H\cap F_{d-1}R}
\]
as a graded subalgebra of $\gr R$.  If $\Cent_{\gr R}(\gr H)=\gr H$, then
\[
        \Cent_R(H)=H.
\]
\end{lemma}

\begin{proof}
Since $H$ is commutative, $H\subseteq\Cent_R(H)$.
Conversely, let $0\neq x\in\Cent_R(H)$, and put
$D=\deg_F x$.  Let $\eta$ be a homogeneous element of $\gr H$, and
choose $h\in H$ such that $\init(h)=\eta$.  Since $[x,h]=0$, its
highest possible homogeneous component gives $[\init(x),\eta]=0$.
Thus
\[
        \init(x)\in
        \Cent_{\gr R}(\gr H)
        =
        \gr H.
\]
There is therefore an element $y\in H\cap F_DR$ satisfying $\init(y)=\init(x)$.
It follows that
\[
        x-y\in F_{D-1}R.
\]
Moreover, $x-y$ still centralises $H$.  Induction on $D$ gives
$x-y\in H$, and hence $x\in H$.
\end{proof}

\subsection{Symmetries and notation}

We shall use three standard symmetries of $\Oq$.  There is an involutive
automorphism $\sigma$, an involutive antiautomorphism $\dagger$, and
an involutive antiautomorphism $\tau$, satisfying, for $n\in\N$,
\begin{align*}
\sigma:&\quad
 W_{-n}\leftrightarrow W_{n+1},
 \qquad
 G_{n+1}\leftrightarrow\widetilde G_{n+1},\\
\dagger:&\quad
 W_{-n}\mapsto W_{-n},
 \qquad
 W_{n+1}\mapsto W_{n+1},
 \qquad
 G_{n+1}\leftrightarrow\widetilde G_{n+1},\\
\tau:&\quad
 W_{-n}\leftrightarrow W_{n+1},
 \qquad
 G_{n+1}\mapsto G_{n+1},
 \qquad
 \widetilde G_{n+1}\mapsto\widetilde G_{n+1}.
\end{align*}
See
\cite[Definitions~3.4--3.6 and Theorem~11.6]
{Terwilliger2022OqACE}.  Each of these maps fixes or interchanges
$W_0,W_1$, and hence preserves the word-length filtration.

From Section~\ref{sec:qshuffle} onwards, the superscript $\sh$ is
reserved for elements of the $q$-shuffle algebra.  Thus
$\cW_{-n}\in\Aq$, $W_{-n}\in\Oq$, and
$W_{-n}^{\sh}\in\Vsh$ are distinguished by their notation.

\section{The \texorpdfstring{$q$}{q}-shuffle algebra}
\label{sec:qshuffle}

In this section, we recall the $q$-shuffle realisation of $\Uplus$, the
alternating words, the Catalan elements, and Damiani's root vectors, using
the normalisations required below.

\subsection{The \texorpdfstring{$q$}{q}-shuffle realisation}

Let $\Vsh$ be the $\F$-vector space with basis the set of all words in
the letters $x,y$, including the empty word $1$.  When words are written
without the symbol $\star$, juxtaposition denotes concatenation.

Let $\langle\, ,\,\rangle$ be the symmetric bilinear form on the free
abelian group $\Z x\oplus\Z y$ determined by
\[
 \langle x,x\rangle=\langle y,y\rangle=2,
 \qquad
 \langle x,y\rangle=\langle y,x\rangle=-2.
\]
The $q$-shuffle product is the unique bilinear product $\star$ on
$\Vsh$ such that $1\star v=v\star1=v$ for $v\in\Vsh$, and, for
non-empty words
$u=u_1\cdots u_r$ and $v=v_1\cdots v_s$,
\begin{equation}\label{eq:qshuffle-recursion}
\begin{aligned}
u\star v
={}
u_1\bigl((u_2\cdots u_r)\star v\bigr)+
v_1\bigl(u\star(v_2\cdots v_s)\bigr)
q^{\sum_{i=1}^{r}\langle u_i,v_1\rangle}.
\end{aligned}
\end{equation}
An empty subword in \eqref{eq:qshuffle-recursion} is interpreted as $1$.
The product $\star$ is associative; see
\cite{Rosso1998} and
\cite[Section~4]{Terwilliger2019Alternating}.  For example,
\[
        x\star y=xy+q^{-2}yx,
        \qquad
        y\star x=yx+q^{-2}xy.
\]

Let $U$ denote the subalgebra of $(\Vsh,\star)$ generated by $x,y$.
Rosso's $q$-shuffle embedding gives an algebra isomorphism
\[
        \iota:\Uplus\xrightarrow{\ \sim\ }U,
        \qquad
        A\longmapsto x,\qquad
        B\longmapsto y;
\]
see \cite{Rosso1998} and
\cite[Definition~4.2 and the paragraph following
equations~(14)--(15)]{Terwilliger2019Alternating}.
We henceforth identify $\Uplus$ with $U$.  Together with
Theorem~\ref{thm:gr-isomorphism}, this allows us to regard initial forms in
$\gr\Oq$ as elements of the $q$-shuffle algebra.

We record the root grading that will be used in
Section~\ref{sec:specialisation}.  Put
$Q_+=\N\alpha_0+\N\alpha_1$ and
$\delta=\alpha_0+\alpha_1$.
The set of positive roots of
$\widehat{\mathfrak{sl}}_2$ is
\[
 \Phi_+
 =
 \{n\delta+\alpha_0\mid n\in\N\}
 \,\sqcup\,
 \{m\delta\mid m\geqslant 1\}
 \,\sqcup\,
 \{n\delta+\alpha_1\mid n\in\N\}.
\]
The algebra $\Uplus$ is $Q_+$-graded by
$\deg A=\alpha_0$ and $\deg B=\alpha_1$.  Under the $q$-shuffle
realisation, $\deg x=\alpha_0$ and $\deg y=\alpha_1$; thus a word
containing $a$ copies of $x$ and $b$ copies of $y$ has degree
$a\alpha_0+b\alpha_1$.

\subsection{Alternating words}

For $n\in\N$, define
\begin{equation}
\label{eq:alternating-words}
 W_{-n}^{\sh}=x(yx)^n,\qquad
 W_{n+1}^{\sh}=y(xy)^n,\qquad
 G_n^{\sh}=(yx)^n,\qquad
 \widetilde G_n^{\sh}=(xy)^n.
\end{equation}
The powers in \eqref{eq:alternating-words} are concatenation powers.  In particular,
\[
 W_0^{\sh}=x,\qquad W_1^{\sh}=y,\qquad
 G_0^{\sh}=\widetilde G_0^{\sh}=1,\qquad
 G_1^{\sh}=yx,\qquad \widetilde G_1^{\sh}=xy.
\]
These are the alternating words of
\cite[Definitions~5.1 and~5.2]{Terwilliger2019Alternating}; the
superscript $\sh$ distinguishes them from the alternating generators of
$\Oq$.

Their degrees are
\begin{align*}
\deg W_{-n}^{\sh}&=n\delta+\alpha_0,&
\deg W_{n+1}^{\sh}&=n\delta+\alpha_1,\\
\deg G_n^{\sh}&=n\delta,&
\deg\widetilde G_n^{\sh}&=n\delta.
\end{align*}

For a word $v=v_1v_2\cdots v_r$, put
$\operatorname{rev}(v)=v_rv_{r-1}\cdots v_1$, and extend this map
linearly to $\Vsh$.  Word reversal is an involutive antiautomorphism of
$(\Vsh,\star)$.  Moreover,
\[
 \operatorname{rev}(W_{-n}^{\sh})=W_{-n}^{\sh},
 \qquad
 \operatorname{rev}(W_{n+1}^{\sh})=W_{n+1}^{\sh},\qquad
 \operatorname{rev}(G_n^{\sh})=\widetilde G_n^{\sh},
 \qquad
 \operatorname{rev}(\widetilde G_n^{\sh})=G_n^{\sh}.
\]
See
\cite[Lemmas~4.1 and~5.3]{Terwilliger2019Alternating}.

Whenever a commutator or $q$-commutator is taken in $\Vsh$, its
products are understood to be $q$-shuffle products.  For $n\geqslant 1$,
\cite[Proposition~5.7, equations~(38)--(39)]
{Terwilliger2019Alternating} gives
\begin{align}
\qcomm{\widetilde G_n^{\sh}}{x}
 &=(q-q^{-1})W_{-n}^{\sh},
\label{eq:qshuffle-Wminus}\\
\qcomm{y}{\widetilde G_n^{\sh}}
 &=(q-q^{-1})W_{n+1}^{\sh}.
\label{eq:qshuffle-Wplus}
\end{align}

The following result is
\cite[Theorems~10.1 and~10.2]{Terwilliger2019Alternating}.

\begin{theorem}
\label{thm:qshuffle-twelve}
The three families
\[
 \{W_{-i}^{\sh}\}_{i\in\N},\qquad
 \{\widetilde G_{j+1}^{\sh}\}_{j\in\N},\qquad
 \{W_{k+1}^{\sh}\}_{k\in\N}
\]
give a PBW basis of $\Uplus$, viewed inside
$(\Vsh,\star)$, for every total order satisfying one of the six block
conditions
\begin{equation}
\label{eq:six-shuffle-orders}
\begin{array}{lll}
W^-<\widetilde G<W^+,&
W^+<\widetilde G<W^-,&
W^+<W^-<\widetilde G,\\[1mm]
W^-<W^+<\widetilde G,&
\widetilde G<W^+<W^-,&
\widetilde G<W^-<W^+.
\end{array}
\end{equation}
No condition is imposed on the order inside a single block.  The same
statement holds with $\widetilde G_{j+1}^{\sh}$ replaced by
$G_{j+1}^{\sh}$ for every $j\in\N$.
\end{theorem}

\subsection{Catalan elements and Damiani root vectors}

Set $\epsilon(x)=1$ and $\epsilon(y)=-1$.
For a word $v=v_1\cdots v_r$, define
\[
        \operatorname{ht}_i(v)
        =
        \sum_{h=1}^{i}\epsilon(v_h)
        \qquad(0\leqslant i\leqslant r),
\]
where $\operatorname{ht}_0(v)=0$.  The word $v$ is called
\textbf{Catalan} if
\[
        \operatorname{ht}_i(v)\geqslant 0
        \qquad(0\leqslant i\leqslant r),
        \qquad
        \operatorname{ht}_r(v)=0.
\]
The empty word is Catalan, and every non-empty Catalan word has even
length.

Let $\mathrm{Cat}_n$ denote the set of Catalan words of length $2n$.
For $n\in\N$, define
\begin{equation}
\label{eq:Catalan-definition}
 C_n
 =
 \sum_{v\in\mathrm{Cat}_n}
 v\prod_{j=0}^{2n}
 \bigl[1+\operatorname{ht}_j(v)\bigr]_q.
\end{equation}
Thus $C_0=1$ and $C_1=[2]_qxy$.  We use the normalisation of
\cite[Definition~1.5]{Terwilliger2019Catalan}; see also
\cite[Section~4]{Terwilliger2019Alternating}.  Every $C_n$ is
homogeneous of degree $n\delta$.

\textbf{Damiani's root vectors}
\cite{Damiani1993} are
\[
 \{E_{n\delta+\alpha_0}\}_{n\in\N},\qquad
 \{E_{n\delta+\alpha_1}\}_{n\in\N},\qquad
 \{E_{n\delta}\}_{n\geqslant 1}.
\]
In the normalisation of
\cite[Section~2, equations~(4)--(6)]
{Terwilliger2019Alternating}, they are defined by
\begin{align}
&E_{\alpha_0}=A,
\qquad\qquad E_{\alpha_1}=B,
\qquad \qquad E_\delta=q^{-2}BA-AB,
\label{eq:Damiani-initial}\\
&E_{n\delta+\alpha_0}
=\frac{\comm{E_\delta}{E_{(n-1)\delta+\alpha_0}}}{[2]_q},
\,\,\,\qquad \qquad E_{n\delta+\alpha_1}
=\frac{\comm{E_{(n-1)\delta+\alpha_1}}{E_\delta}}{[2]_q}
\qquad (n\geqslant 1),
\label{eq:Damiani-real}\\
&E_{n\delta}
=q^{-2}E_{(n-1)\delta+\alpha_1}A
-AE_{(n-1)\delta+\alpha_1}
\qquad (n\geqslant 1).
\label{eq:Damiani-imaginary}
\end{align}
The denominator $[2]_q$ is non-zero by
Lemma~\ref{lem:nonvanishing-scalars}.  Each root vector $E_\beta$ is
homogeneous of degree $\beta$.

The $q$-shuffle images of these root vectors are given by
\cite[Theorem~1.7]{Terwilliger2019Catalan}:
\begin{align}
\iota(E_{n\delta+\alpha_0})
 &=
 q^{-2n}(q-q^{-1})^{2n}xC_n
 \qquad(n\in\N),
\label{eq:Damiani-shuffle-0}\\
\iota(E_{n\delta+\alpha_1})
 &=
 q^{-2n}(q-q^{-1})^{2n}C_ny
 \qquad(n\in\N),
\label{eq:Damiani-shuffle-1}\\
\iota(E_{n\delta})
 &=
 -q^{-2n}(q-q^{-1})^{2n-1}C_n
 \qquad(n\geqslant 1).
\label{eq:Damiani-shuffle-delta}
\end{align}
In \eqref{eq:Damiani-shuffle-0} and
\eqref{eq:Damiani-shuffle-1}, the expressions $xC_n$ and $C_ny$
denote concatenation, not $q$-shuffle multiplication.

We shall use the following identity in
Section~\ref{sec:alternating-initial}.  It is
\cite[Theorem~11.14]{Terwilliger2019Alternating}.

\begin{theorem}
\label{thm:Catalan-alternating}
For every $n\in\N$,
\begin{equation}
\label{eq:Catalan-alternating}
        \sum_{i=0}^{n}
        (-1)^i[2n-i]_q\,
        C_i\star\widetilde G_{n-i}^{\sh}
        =0.
\end{equation}
\end{theorem}

\section{Damiani PBW bases in arbitrary order}
\label{sec:specialisation}

We prove that Damiani's root vectors give a PBW basis of $\Uplus$ for
every total order on $\Phi_+$.  The proof uses the explicit straightening
relations and is carried out directly over $\F$ under the standing
assumption that $q$ is not a root of unity.

For $i\in\N$ and $m\geqslant 1$, write
$A_i=E_{i\delta+\alpha_0}$,
$B_i=E_{i\delta+\alpha_1}$, and $I_m=E_{m\delta}$.
Let $<_{\mathrm D}$ denote the total order determined by
\begin{equation}
\label{eq:Damiani-standard-order}
 A_i<_{\mathrm D}A_{i+1},\qquad
 I_m<_{\mathrm D}I_{m+1},\qquad
 B_{j+1}<_{\mathrm D}B_j,\qquad
 A_i<_{\mathrm D}I_m<_{\mathrm D}B_j
\end{equation}
for $i,j\in\N$ and $m\geqslant 1$.
The $<_{\mathrm D}$-ordered monomials form an $\F$-basis of
$\Uplus$.  Under the present assumptions on $\F$ and $q$, this is
\cite[Proposition~2.4]{Terwilliger2019Alternating}; the original result is
\cite[Section~5, Theorem~2, p.~308]{Damiani1993}.

For a word $w$ in the root vectors, let $a_i(w)$ and $b_i(w)$ be
the multiplicities of $A_i$ and $B_i$, respectively.  Define the
\textbf{real signature} of $w$ by
\begin{equation}
\label{eq:real-signature}
 \operatorname{rsig}(w)
 =
 \bigl(
 a_0(w),a_1(w),a_2(w),\ldots;
 b_0(w),b_1(w),b_2(w),\ldots
 \bigr).
\end{equation}
Only finitely many entries are non-zero.  We order the signatures
lexicographically, reading the $A$-coordinates first and then the
$B$-coordinates, with the larger entry at the first place of
disagreement giving the larger signature.  For each fixed weight
$\nu\in Q_+$, only finitely many real signatures occur.

The real signature is additive under concatenation.  In particular, for
words $u,v,w,w'$,
\begin{equation}
\label{eq:signature-additivity}
 \operatorname{rsig}(w)<\operatorname{rsig}(w')
 \quad\Longrightarrow\quad
 \operatorname{rsig}(uwv)<\operatorname{rsig}(uw'v).
\end{equation}

We shall use the following straightening relations.  For
$i,j\in\N$ and $m\geqslant 1$,
\begin{align}
 B_iA_j
 &=
 q^2A_jB_i+q^2I_{i+j+1},
 \label{eq:Damiani-cross-BA}\\
 I_mA_j
 &=
 A_jI_m+q^{2-2m}[2]_qA_{m+j}
 -q^2(q^2-q^{-2})
 \sum_{\ell=1}^{m-1}
 q^{-2\ell}A_{j+\ell}I_{m-\ell},
 \label{eq:Damiani-cross-IA}\\
 B_jI_m
 &=
 I_mB_j+q^{2-2m}[2]_qB_{m+j}
 -q^2(q^2-q^{-2})
 \sum_{\ell=1}^{m-1}
 q^{-2\ell}I_{m-\ell}B_{j+\ell}.
 \label{eq:Damiani-cross-BI}
\end{align}
Let $i>j\geqslant 0$.  If $i-j=2r+1$, then
\begin{align}
 A_iA_j
 &=
 q^{-2}A_jA_i
 -(q^2-q^{-2})
 \sum_{\ell=1}^{r}
 q^{-2\ell}A_{j+\ell}A_{i-\ell},
 \label{eq:Damiani-cross-AA-odd}\\
 B_jB_i
 &=
 q^{-2}B_iB_j
 -(q^2-q^{-2})
 \sum_{\ell=1}^{r}
 q^{-2\ell}B_{i-\ell}B_{j+\ell}.
 \label{eq:Damiani-cross-BB-odd}
\end{align}
If $i-j=2r$, then
\begin{align}
 A_iA_j
 ={}&
 q^{-2}A_jA_i
 -q^{j-i+1}(q-q^{-1})A_{j+r}^{\,2}-
 (q^2-q^{-2})
 \sum_{\ell=1}^{r-1}
 q^{-2\ell}A_{j+\ell}A_{i-\ell},
 \label{eq:Damiani-cross-AA-even}\\
 B_jB_i
 ={}&
 q^{-2}B_iB_j
 -q^{j-i+1}(q-q^{-1})B_{j+r}^{\,2}-
 (q^2-q^{-2})
 \sum_{\ell=1}^{r-1}
 q^{-2\ell}B_{i-\ell}B_{j+\ell}.
 \label{eq:Damiani-cross-BB-even}
\end{align}
Empty sums are zero, and
\begin{equation}
\label{eq:Damiani-cross-II}
        I_mI_n=I_nI_m
        \qquad(m,n\geqslant 1).
\end{equation}

The relations
\eqref{eq:Damiani-cross-BA}--\eqref{eq:Damiani-cross-BB-even}
are recorded, under the present assumptions on the field and the
parameter, in
\cite[Lemmas~3.1--3.3]{Terwilliger2019Catalan}.
Their original sources are
\cite[p.~307]{Damiani1993} for
\eqref{eq:Damiani-cross-BA},
\cite[p.~304]{Damiani1993} for
\eqref{eq:Damiani-cross-IA}--\eqref{eq:Damiani-cross-BI}, and
\cite[p.~300]{Damiani1993} for
\eqref{eq:Damiani-cross-AA-odd}--\eqref{eq:Damiani-cross-BB-even}.
The commutativity in \eqref{eq:Damiani-cross-II} is
\cite[p.~307]{Damiani1993}.

\begin{lemma}
\label{lem:Damiani-triangular-crossings}
Let $X,Y$ be distinct Damiani root vectors.  There are
$\varepsilon_{X,Y}\in\{1,q^2,q^{-2}\}$, finitely many scalars
$c_s\in\F$, and words $w_s$ in the Damiani root vectors such that
\begin{equation}
\label{eq:Damiani-triangular-crossing}
        XY
        =
        \varepsilon_{X,Y}YX+\sum_s c_sw_s,
\end{equation}
where every $w_s$ has the same weight as $XY$ and satisfies
\[
        \operatorname{rsig}(w_s)<\operatorname{rsig}(XY).
\]
The same conclusion holds after solving
\eqref{eq:Damiani-triangular-crossing} for $YX$.
\end{lemma}

\begin{proof}
All the displayed straightening relations are homogeneous, so every
correction word has the same weight as the product being straightened.

In \eqref{eq:Damiani-cross-BA}, the correction term contains no real root
vector, whereas $B_iA_j$ contains both $A_j$ and $B_i$.  Its real
signature is therefore smaller.

In \eqref{eq:Damiani-cross-IA}, each correction replaces $A_j$ by an
$A$-factor of index strictly larger than $j$.  Hence the first
$A$-coordinate at which the signature changes is the $j$-th
coordinate, and this coordinate decreases.  The same argument, using the
$B$-coordinates, applies to \eqref{eq:Damiani-cross-BI}.

Consider
\eqref{eq:Damiani-cross-AA-odd}--\eqref{eq:Damiani-cross-BB-even}.
Every correction replaces the pair of indices $j<i$ by two indices
strictly larger than $j$; in the even case the midpoint occurs twice.
Thus the multiplicity at index $j$ decreases, and no earlier real
coordinate changes.  Each correction again has smaller real signature.

The leading coefficients are $1,q^2,q^{-2}$, all of which are non-zero.
Solving a relation in the opposite direction merely multiplies its
correction terms by a non-zero scalar.  Since
$\operatorname{rsig}(XY)=\operatorname{rsig}(YX)$, the strict signature
inequality is unchanged.
\end{proof}

\begin{theorem}
\label{thm:any-Damiani}
Let $\prec$ be any total order on $\Phi_+$.  The
$\prec$-ordered monomials in Damiani's root vectors form an
$\F$-basis of $\Uplus$.
\end{theorem}

\begin{proof}
Fix $\nu\in Q_+$, and let $U_\nu^+$ denote the homogeneous component
of $\Uplus$ of weight $\nu$.  Since $A_0=A$ and $B_0=B$, it is
enough to straighten words in the Damiani root vectors.

For a word $w=X_1X_2\cdots X_d$ of weight $\nu$, let
\[
 \operatorname{inv}_{\prec}(w)
 =
 \bigl|\{(r,s)\mid 1\leqslant r<s\leqslant d,\;
                    X_s\prec X_r\}\bigr|.
\]
We use lexicographic induction on
$\bigl(\operatorname{rsig}(w),\operatorname{inv}_{\prec}(w)\bigr)$,
with the real signature as the first coordinate.  This induction is
well-founded, since only finitely many real signatures occur in weight
$\nu$, while the inversion number is a non-negative integer.

Suppose that $w$ is not $\prec$-ordered.  Then $w$ contains an
adjacent inversion $XY$, with $Y\prec X$.  Write $w=uXYv$.
By Lemma~\ref{lem:Damiani-triangular-crossings},
\[
        uXYv
        =
        \varepsilon_{X,Y}uYXv
        +
        \sum_s c_suw_sv.
\]
The word $uYXv$ has the same real signature as $w$ and exactly one
fewer $\prec$-inversion.  For every $s$,
\eqref{eq:signature-additivity} gives
$\operatorname{rsig}(uw_sv)<\operatorname{rsig}(uXYv)$.
The induction hypothesis therefore expresses every term on the
right-hand side as a linear combination of $\prec$-ordered monomials.
It follows that the $\prec$-ordered monomials span $U_\nu^+$.

Let
\[
 \mathcal P(\nu)
 =
 \left\{
 e:\Phi_+\longrightarrow\N
 \ \middle|\
 \begin{array}{l}
 e\text{ is finitely supported},\\[1mm]
 \displaystyle
 \sum_{\beta\in\Phi_+}e(\beta)\beta=\nu
 \end{array}
 \right\}.
\]
The set $\mathcal P(\nu)$ is finite.  For every
$e\in\mathcal P(\nu)$, there is exactly one $\prec$-ordered monomial
with root multiplicities $e$.  Thus the spanning family just obtained is
indexed by $\mathcal P(\nu)$.

The $<_{\mathrm D}$-ordered Damiani monomials of weight $\nu$ are
indexed by the same set $\mathcal P(\nu)$, and they form a basis.
Consequently,
$\dim_{\F}U_\nu^+=|\mathcal P(\nu)|$.  The $\prec$-ordered spanning
family therefore consists of exactly $\dim_{\F}U_\nu^+$ vectors, and hence is a basis of $U_\nu^+$.
Taking the direct sum over $\nu\in Q_+$ proves the result.
\end{proof}

\section{Baseilhac--Kolb root-vector PBW bases}
\label{sec:root-initial}

We compute the initial forms of the Baseilhac--Kolb root vectors and use
them to lift the arbitrary-order Damiani PBW bases from
$\gr\Oq\simeq\Uplus$ to $\Oq$.

The \textbf{Baseilhac--Kolb root vectors} are
\[
 \{B_{n\delta+\alpha_0}\}_{n\in\N},\qquad
 \{B_{n\delta+\alpha_1}\}_{n\in\N},\qquad
 \{B_{n\delta}\}_{n\geqslant 1}.
\]
In the normalisation used here, they are defined as follows:
\begin{equation}
\label{eq:Bdelta}
        B_\delta=q^{-2}W_1W_0-W_0W_1,
\end{equation}
and
\begin{align}
&B_{\alpha_0}
 =W_0, \qquad \qquad
B_{\delta+\alpha_0}
 =W_1+
   \frac{q\,\comm{B_\delta}{W_0}}
        {(q-q^{-1})(q^2-q^{-2})},
\label{eq:BK-real0-first}\\
&B_{n\delta+\alpha_0}
 =B_{(n-2)\delta+\alpha_0}
   +\frac{q\,\comm{B_\delta}{B_{(n-1)\delta+\alpha_0}}}
        {(q-q^{-1})(q^2-q^{-2})}
        \qquad(n\geqslant 2),
\label{eq:BK-real0}\\
&B_{\alpha_1}
 =W_1, \qquad\qquad
B_{\delta+\alpha_1}
 =W_0-
   \frac{q\,\comm{B_\delta}{W_1}}
        {(q-q^{-1})(q^2-q^{-2})},
\label{eq:BK-real1-first}\\
&B_{n\delta+\alpha_1}
 =B_{(n-2)\delta+\alpha_1}
   -\frac{q\,\comm{B_\delta}{B_{(n-1)\delta+\alpha_1}}}
        {(q-q^{-1})(q^2-q^{-2})}
        \qquad(n\geqslant 2).
\label{eq:BK-real1}
\end{align}
For $n\geqslant 1$, the imaginary root vectors are defined by
\begin{align}
B_{n\delta}
={}&
q^{-2}B_{(n-1)\delta+\alpha_1}W_0
-W_0B_{(n-1)\delta+\alpha_1}
+(q^{-2}-1)
\sum_{\ell=0}^{n-2}
B_{\ell\delta+\alpha_1}
B_{(n-\ell-2)\delta+\alpha_1}.
\label{eq:BK-imaginary}
\end{align}
An empty sum is zero.  These formulae are
\cite[Section~4, equations~(7)--(12)]{Terwilliger2022OqACE};
see also \cite[Section~3]{BaseilhacKolb2020}.

We shall use
\begin{equation}
\label{eq:denominator-expanded}
 (q-q^{-1})(q^2-q^{-2})
 =
 (q-q^{-1})^2[2]_q.
\end{equation}
All denominators above are non-zero by
Lemma~\ref{lem:nonvanishing-scalars}.

A filtered comparison of the Baseilhac--Kolb and Damiani root vectors over
a generic coefficient field appears in
\cite[Proposition~4.4]{BaseilhacKolb2020}.  The next proposition gives the
exact scalar factors in the present normalisation and under our standing
assumption on $q$.

\begin{proposition}
\label{prop:root-leading-U}
For $n\in\N$,
\begin{align}
\init(B_{n\delta+\alpha_0})
 &=
 q^n(q-q^{-1})^{-2n}
 E_{n\delta+\alpha_0},
\label{eq:root-leading-real0}\\
\init(B_{n\delta+\alpha_1})
 &=
 q^n(q-q^{-1})^{-2n}
 E_{n\delta+\alpha_1}.
\label{eq:root-leading-real1}
\end{align}
For $n\geqslant 1$,
\begin{equation}
\label{eq:root-leading-imaginary}
\init(B_{n\delta})
 =
 q^{n-1}(q-q^{-1})^{2-2n}E_{n\delta}.
\end{equation}
Consequently,
\[
 \deg_F B_{n\delta+\alpha_0}
 =
 \deg_F B_{n\delta+\alpha_1}
 =
 2n+1,
 \qquad
 \deg_F B_{n\delta}=2n.
\]
\end{proposition}

\begin{proof}
Put $c_n=q^n(q-q^{-1})^{-2n}$ for $n\in\N$.  Every Damiani root
vector is non-zero by
Theorem~\ref{thm:any-Damiani}.

Equation~\eqref{eq:Bdelta} gives
$[B_\delta]_2=q^{-2}BA-AB=E_\delta$.  In particular, $B_\delta$
has filtration degree $2$.

We first prove \eqref{eq:root-leading-real0} by strong induction on $n$.
For $n=0$, one has
$[B_{\alpha_0}]_1=A=E_{\alpha_0}=c_0E_{\alpha_0}$.
Let $n\geqslant 1$, and suppose that the assertion holds at all smaller
indices.  In the defining recursion for
$B_{n\delta+\alpha_0}$, the term not involving a commutator has
filtration degree strictly smaller than $2n+1$.  The recursion therefore
gives $B_{n\delta+\alpha_0}\in F_{2n+1}\Oq$, and
\begin{align*}
[B_{n\delta+\alpha_0}]_{2n+1}
&=
\frac{qc_{n-1}}
     {(q-q^{-1})(q^2-q^{-2})}
\comm{E_\delta}{E_{(n-1)\delta+\alpha_0}}\\
&=
\frac{qc_{n-1}[2]_q}
     {(q-q^{-1})^2[2]_q}
E_{n\delta+\alpha_0}\\
&=
c_nE_{n\delta+\alpha_0},
\end{align*}
where we used \eqref{eq:denominator-expanded} and
\eqref{eq:Damiani-real}.  This class is non-zero, proving both the
filtration degree and \eqref{eq:root-leading-real0}.

The same strong induction proves
\eqref{eq:root-leading-real1}.  The minus sign in
\eqref{eq:BK-real1-first} and \eqref{eq:BK-real1} reverses the commutator,
as required by the second relation in \eqref{eq:Damiani-real}.  Hence
$[B_{n\delta+\alpha_1}]_{2n+1}=c_nE_{n\delta+\alpha_1}$.

It remains to consider $B_{n\delta}$.  By the real-root formula,
$[B_{(n-1)\delta+\alpha_1}]_{2n-1}
=c_{n-1}E_{(n-1)\delta+\alpha_1}$.
Hence the first two terms of \eqref{eq:BK-imaginary} belong to
$F_{2n}\Oq$, and their degree-$2n$ component is
\begin{align*}
c_{n-1}
\bigl(
q^{-2}E_{(n-1)\delta+\alpha_1}A
-
AE_{(n-1)\delta+\alpha_1}
\bigr)
&=
c_{n-1}E_{n\delta}=
q^{n-1}(q-q^{-1})^{2-2n}E_{n\delta},
\end{align*}
by \eqref{eq:Damiani-imaginary}.

Every summand in the final sum of \eqref{eq:BK-imaginary} has filtration
degree at most
$(2\ell+1)+(2(n-\ell-2)+1)=2n-2$.  It therefore makes no contribution in degree $2n$.  Since the displayed
degree-$2n$ component is non-zero,
\eqref{eq:root-leading-imaginary} follows.
\end{proof}

Substituting the $q$-shuffle formulae for the Damiani root vectors gives
the following form of Proposition~\ref{prop:root-leading-U}.

\begin{corollary}
\label{cor:root-leading-shuffle}
For $n\in\N$,
\begin{align}
\init(B_{n\delta+\alpha_0})
 &=q^{-n}xC_n,
\label{eq:BK-leading-xC}\\
\init(B_{n\delta+\alpha_1})
 &=q^{-n}C_ny.
\label{eq:BK-leading-Cy}
\end{align}
For $n\geqslant 1$,
\begin{equation}
\label{eq:BK-leading-C}
        \init(B_{n\delta})
        =
        -q^{-n-1}(q-q^{-1})C_n.
\end{equation}
\end{corollary}

\begin{proof}
Substitute
\eqref{eq:Damiani-shuffle-0}--\eqref{eq:Damiani-shuffle-delta}
into
\eqref{eq:root-leading-real0}--\eqref{eq:root-leading-imaginary}
and simplify the scalar factors.
\end{proof}

For later use, we make the scalar convention
\begin{equation}
\label{eq:B0delta-convention}
        B_{0\delta}
        =
        q^{-2}-1
        =
        -q^{-1}(q-q^{-1}).
\end{equation}
Since $C_0=1$, equation~\eqref{eq:BK-leading-C} remains valid for
$n=0$ under this convention.  The element $B_{0\delta}$ in
\eqref{eq:B0delta-convention} is a scalar, not an additional root vector.

For $\beta\in\Phi_+$, let $B_\beta$ denote the corresponding
Baseilhac--Kolb root vector.

\begin{theorem}
\label{thm:BK-nonroot}
Let $\prec$ be any total order on $\Phi_+$.  The
$\prec$-ordered monomials in the Baseilhac--Kolb root vectors
$\{B_\beta\mid\beta\in\Phi_+\}$ form an $\F$-basis of $\Oq$.
\end{theorem}

\begin{proof}
By Proposition~\ref{prop:root-leading-U}, for every
$\beta\in\Phi_+$ there is a scalar $a_\beta\in\F^\times$ such that
$\init(B_\beta)=a_\beta E_\beta$.
The scalar is non-zero by Lemma~\ref{lem:nonvanishing-scalars}, and every
$B_\beta$ has positive filtration degree.

By Theorem~\ref{thm:any-Damiani}, the $\prec$-ordered monomials in the
Damiani root vectors $E_\beta$ form a basis of
$\Uplus\simeq\gr\Oq$.  Lemma~\ref{lem:filtered-lifting} therefore lifts
this basis to the $\prec$-ordered monomials in the root vectors
$B_\beta$.
\end{proof}

Theorem~\ref{thm:BK-nonroot} proves
\cite[Conjecture~16.1]{Terwilliger2022OqACE}.

\section{Initial forms of the alternating generators}
\label{sec:alternating-initial}

In this section, we determine the initial forms of the four families of
alternating generators of $\Oq$.  We use the identifications
$\gr\Oq\simeq\Uplus\subseteq(\Vsh,\star)$ fixed in
Sections~\ref{sec:prelim} and~\ref{sec:qshuffle}.  Accordingly,
products of initial forms below are written as $q$-shuffle products.

We retain the scalar convention
\begin{equation}
\label{eq:G0-Oq}
        G_0=\widetilde G_0
        =-(q-q^{-1})[2]_q^2.
\end{equation}
Together with the convention \eqref{eq:B0delta-convention}, this gives
\begin{equation}
\label{eq:B-leading-all-n}
 [B_{m\delta}]_{2m}
 =
 -q^{-m-1}(q-q^{-1})C_m
 \qquad(m\in\N).
\end{equation}

The following relations are the cases $k=n-1$ of
\cite[Theorem~11.6, equations~(110)--(111)]
{Terwilliger2022OqACE}.  For $n\geqslant 1$,
\begin{align}
\qcomm{\widetilde G_n}{W_0}
 &=
 \rho_0W_{-n}-\rho_0W_n,
\label{eq:Oq-tG-W0}\\
\qcomm{W_1}{\widetilde G_n}
 &=
 \rho_0W_{n+1}-\rho_0W_{-(n-1)}.
\label{eq:Oq-W1-tG}
\end{align}
We shall also use
\cite[Lemma~11.8]{Terwilliger2022OqACE}, which gives
\begin{equation}
\label{eq:dagger-alternating}
        G_n^\dagger=\widetilde G_n,
        \qquad
        \widetilde G_n^\dagger=G_n
        \qquad(n\geqslant 1).
\end{equation}

\subsection{The family \texorpdfstring{$\widetilde G$}{G-tilde}}

The recursion relating the imaginary Baseilhac--Kolb root vectors to the
family $\{\widetilde G_n\}_{n\in\N}$ is
\cite[Lemma~12.5, equation~(133)]{Terwilliger2022OqACE}.  For
$n\geqslant 1$,
\begin{equation}
\label{eq:root-tG-recursion}
\begin{aligned}
0={}&
[n]_qB_{n\delta}\widetilde G_0+
\sum_{\substack{j+k+2\ell+1=n\\j,k,\ell\geqslant 0}}
(-1)^\ell
\binom{k+\ell}{\ell}
[2n-j]_q[2]_q^{k+1}
B_{j\delta}\widetilde G_{k+1}.
\end{aligned}
\end{equation}

\begin{proposition}
\label{prop:tG-leading}
For every $n\in\N$,
\begin{equation}
\label{eq:tG-leading}
\init(\widetilde G_n)
 =
 (-1)^{n+1}(q-q^{-1})q^{-n}[2]_q^{2-n}
 \widetilde G_n^{\sh}.
\end{equation}
In particular, $\deg_F\widetilde G_n=2n$.
\end{proposition}

\begin{proof}
For $r\in\N$, put
\[
 X_r
 =
 (-1)^{r+1}(q-q^{-1})q^{-r}[2]_q^{2-r}
 \widetilde G_r^{\sh}.
\]
For $r=0$, equation~\eqref{eq:G0-Oq} gives
$\init(\widetilde G_0)=X_0$.

Let $n\geqslant 1$, and assume that
$\init(\widetilde G_r)=X_r$ for $0\leqslant r<n$.
In \eqref{eq:root-tG-recursion}, the unique term containing
$\widetilde G_n$ corresponds to
$(j,k,\ell)=(0,n-1,0)$.  Its coefficient is
$d_n=[2n]_q[2]_q^nB_{0\delta}\in\F^\times$.
Indeed, each factor is non-zero by
Lemma~\ref{lem:nonvanishing-scalars} and
\eqref{eq:B0delta-convention}.

The first term in \eqref{eq:root-tG-recursion} belongs to
$F_{2n}\Oq$.  Every summand other than the one containing
$\widetilde G_n$ involves an alternating generator of index smaller
than $n$, and belongs to filtration degree at most
$2j+2(k+1)=2n-4\ell\leqslant 2n$.  Since $d_n\neq 0$, solving
\eqref{eq:root-tG-recursion} for $\widetilde G_n$ gives
$\widetilde G_n\in F_{2n}\Oq$.

For $m\in\N$, set $b_m=[B_{m\delta}]_{2m}$.
By \eqref{eq:B-leading-all-n},
\begin{equation}
\label{eq:bm-Cm}
        b_m=-q^{-m-1}(q-q^{-1})C_m.
\end{equation}
Taking the degree-$2n$ component of
\eqref{eq:root-tG-recursion}, the terms with $\ell>0$ disappear.
Writing $r=k+1$, and including the first term of
\eqref{eq:root-tG-recursion} as the term $r=0$, we obtain
\begin{equation}
\label{eq:top-recursion-Oq}
0=
\sum_{r=0}^{n}
[n+r]_q[2]_q^r\,
b_{n-r}\star[\widetilde G_r]_{2r}.
\end{equation}

For $0\leqslant r\leqslant n$, equations
\eqref{eq:bm-Cm} and the definition of $X_r$ give
\begin{align}
[2]_q^r b_{n-r}\star X_r
={}&
(-1)^r q^{-n-1}(q-q^{-1})^2[2]_q^2
 C_{n-r}\star\widetilde G_r^{\sh}.
\label{eq:top-recursion-substitution}
\end{align}
Consequently,
\begin{align*}
\sum_{r=0}^{n}
[n+r]_q[2]_q^r b_{n-r}\star X_r
={}
q^{-n-1}(q-q^{-1})^2[2]_q^2\times
\sum_{r=0}^{n}
(-1)^r[n+r]_q
C_{n-r}\star\widetilde G_r^{\sh}.
\end{align*}
The final sum is zero by
Theorem~\ref{thm:Catalan-alternating}: substitute $i=n-r$ in
\eqref{eq:Catalan-alternating} and multiply by $(-1)^n$.

Thus $X_n$ satisfies \eqref{eq:top-recursion-Oq}.  The coefficient of
the unknown class $[\widetilde G_n]_{2n}$ in that equation is
$[2n]_q[2]_q^n b_0=[2n]_q[2]_q^nB_{0\delta}$, which is non-zero.
Hence the solution is unique, and $[\widetilde G_n]_{2n}=X_n$.
The right-hand side is non-zero, since
$\widetilde G_n^{\sh}=(xy)^n$ is a non-zero word.  Therefore
$\deg_F \widetilde G_n=2n$, and
\eqref{eq:tG-leading} follows.
\end{proof}

The first imaginary generators admit particularly simple expressions.

\begin{lemma}
\label{lem:first-imaginary-generators}
One has
\begin{align}
\widetilde G_1
 &=-qB_\delta
  =\qcomm{W_0}{W_1},
\label{eq:tG1-Bdelta}\\
G_1
 &=\qcomm{W_1}{W_0}.
\label{eq:G1-low-check}
\end{align}
\end{lemma}

\begin{proof}
Taking $n=1$ in \eqref{eq:root-tG-recursion} gives
\[
        0=B_\delta\widetilde G_0
        +[2]_q^2B_{0\delta}\widetilde G_1.
\]
Substituting \eqref{eq:G0-Oq} and
\eqref{eq:B0delta-convention}, and cancelling the resulting non-zero
scalar, gives $\widetilde G_1=-qB_\delta$.
Equation \eqref{eq:Bdelta} now yields
\[
        -qB_\delta
        =
        qW_0W_1-q^{-1}W_1W_0
        =
        \qcomm{W_0}{W_1}.
\]
Applying $\dagger$ and using \eqref{eq:dagger-alternating} gives
\eqref{eq:G1-low-check}.
\end{proof}

\subsection{The families \texorpdfstring{$W^-$}{W-} and
\texorpdfstring{$W^+$}{W+}}

\begin{proposition}
\label{prop:W-leading}
For every $n\in\N$,
\begin{align}
\init(W_{-n})
 &=
 (-1)^nq^{-n}[2]_q^{-n}W_{-n}^{\sh},
\label{eq:Wminus-leading}\\
\init(W_{n+1})
 &=
 (-1)^nq^{-n}[2]_q^{-n}W_{n+1}^{\sh}.
\label{eq:Wplus-leading}
\end{align}
In particular,
$\deg_F W_{-n}=\deg_F W_{n+1}=2n+1$.
\end{proposition}

\begin{proof}
We argue by induction on $n$.  For $n=0$, one has
$\init(W_0)=x=W_0^{\sh}$ and
$\init(W_1)=y=W_1^{\sh}$.

Let $n\geqslant 1$, and assume the result at all smaller indices.
By Proposition~\ref{prop:tG-leading},
$\widetilde G_n\in F_{2n}\Oq$.  Equations
\eqref{eq:Oq-tG-W0} and \eqref{eq:Oq-W1-tG}, together with the induction
hypothesis, imply $W_{-n},W_{n+1}\in F_{2n+1}\Oq$.
Taking the degree-$(2n+1)$ components in those relations gives
\begin{align}
\qcomm{\init(\widetilde G_n)}{x}
 &=
 \rho_0[W_{-n}]_{2n+1},
\label{eq:top-Wminus}\\
\qcomm{y}{\init(\widetilde G_n)}
 &=
 \rho_0[W_{n+1}]_{2n+1}.
\label{eq:top-Wplus}
\end{align}
The terms $W_n$ and $W_{-(n-1)}$ disappear because they have
filtration degree $2n-1$.

By Proposition~\ref{prop:tG-leading} and
\eqref{eq:qshuffle-Wminus}--\eqref{eq:qshuffle-Wplus}, the left-hand
sides of \eqref{eq:top-Wminus} and \eqref{eq:top-Wplus} are
respectively
\begin{align*}
(-1)^{n+1}(q-q^{-1})^2q^{-n}[2]_q^{2-n}W_{-n}^{\sh}
\quad\text{and}\quad
(-1)^{n+1}(q-q^{-1})^2q^{-n}[2]_q^{2-n}W_{n+1}^{\sh}.
\end{align*}
Since $\rho_0=-(q-q^{-1})^2[2]_q^2$,
division by $\rho_0$ gives
\begin{align*}
[W_{-n}]_{2n+1}
 &=
 (-1)^nq^{-n}[2]_q^{-n}W_{-n}^{\sh},\\
[W_{n+1}]_{2n+1}
 &=
 (-1)^nq^{-n}[2]_q^{-n}W_{n+1}^{\sh}.
\end{align*}
Both classes are non-zero, so the asserted filtration degrees and initial
forms follow.
\end{proof}

We now collect the four initial-form formulae.

\begin{theorem}
\label{thm:all-leading}
For every $n\in\N$,
\begin{align}
\init(W_{-n})
 &=
 (-1)^nq^{-n}[2]_q^{-n}W_{-n}^{\sh},
\label{eq:all-leading-Wminus}\\
\init(W_{n+1})
 &=
 (-1)^nq^{-n}[2]_q^{-n}W_{n+1}^{\sh},
\label{eq:all-leading-Wplus}\\
\init(G_n)
 &=
 (-1)^{n+1}(q-q^{-1})q^{-n}[2]_q^{2-n}
 G_n^{\sh},
\label{eq:all-leading-G}\\
\init(\widetilde G_n)
 &=
 (-1)^{n+1}(q-q^{-1})q^{-n}[2]_q^{2-n}
 \widetilde G_n^{\sh}.
\label{eq:all-leading-tG}
\end{align}
In particular,
\[
 \deg_F W_{-n}=\deg_F W_{n+1}=2n+1,
 \qquad
 \deg_F G_n=\deg_F\widetilde G_n=2n.
\]
\end{theorem}

\begin{proof}
Equations \eqref{eq:all-leading-Wminus} and
\eqref{eq:all-leading-Wplus} follow from
Proposition~\ref{prop:W-leading}, while
\eqref{eq:all-leading-tG} follows from
Proposition~\ref{prop:tG-leading}.

For $n=0$, equation \eqref{eq:all-leading-G} is
\eqref{eq:G0-Oq}.  Let $n\geqslant 1$.  The antiautomorphism
$\dagger$ preserves the word-length filtration, since it fixes
$W_0,W_1$.  Its induced antiautomorphism on $\gr\Oq$ fixes
$A,B$.  Under the $q$-shuffle realisation, it therefore agrees with
the word-reversal antiautomorphism on the subalgebra generated by $x,y$.
Using \eqref{eq:dagger-alternating}, we obtain
\begin{align*}
\init(G_n)
 &=
 \operatorname{rev}\bigl(\init(\widetilde G_n)\bigr)\\
 &=
 (-1)^{n+1}(q-q^{-1})q^{-n}[2]_q^{2-n}
 \operatorname{rev}(\widetilde G_n^{\sh})\\
 &=
 (-1)^{n+1}(q-q^{-1})q^{-n}[2]_q^{2-n}
 G_n^{\sh}.
\end{align*}
This proves \eqref{eq:all-leading-G}.  Each scalar appearing in
\eqref{eq:all-leading-Wminus}--\eqref{eq:all-leading-tG} is non-zero, so
the stated filtration degrees follow.
\end{proof}

\section{Alternating PBW bases}
\label{sec:twelve}

We now combine the initial-form calculations of
Theorem~\ref{thm:all-leading} with the alternating PBW bases of
Theorem~\ref{thm:qshuffle-twelve}.

Set
\[
 W^-=\{W_{-i}\mid i\in\N\},\qquad
 W^+=\{W_{i+1}\mid i\in\N\},\qquad
 G=\{G_{i+1}\mid i\in\N\},\qquad
 \widetilde G=\{\widetilde G_{i+1}\mid i\in\N\}.
\]
Thus neither $G$ nor $\widetilde G$ contains the scalar element with
index zero.

Let $\prec$ be a total order on a union of three families.  We write
$X<Y<Z$ if every element of $X$ precedes every element of $Y$, and
every element of $Y$ precedes every element of $Z$, with respect to
$\prec$.  This notation imposes no condition on the order within an
individual family.

\begin{theorem}
\label{thm:twelve-PBW}
Let $H\in\{G,\widetilde G\}$, and let $\prec$ be a total order on
$W^-\cup H\cup W^+$.
Suppose that $\prec$ satisfies one of the six block conditions
\begin{equation}
\label{eq:six-Oq-orders}
\begin{array}{lll}
W^-<H<W^+,&
W^+<H<W^-,&
W^+<W^-<H,\\[1mm]
W^-<W^+<H,&
H<W^+<W^-,&
H<W^-<W^+.
\end{array}
\end{equation}
Then the $\prec$-ordered monomials in
$W^-\cup H\cup W^+$ form an $\F$-basis of $\Oq$.
\end{theorem}

\begin{proof}
For every $X\in W^-\cup H\cup W^+$, let $X^{\sh}$ denote the
corresponding alternating word in the $q$-shuffle algebra.  By
Theorem~\ref{thm:all-leading}, $X$ has positive filtration degree and
$\init(X)=c_XX^{\sh}$ for some $c_X\in\F^\times$.

Transfer the order $\prec$ to the corresponding alternating words.
By Theorem~\ref{thm:qshuffle-twelve}, the $\prec$-ordered monomials in
\[
        \{X^{\sh}\mid X\in W^-\cup H\cup W^+\}
\]
form an $\F$-basis of $\Uplus\simeq\gr\Oq$.
Lemma~\ref{lem:filtered-lifting} now shows that the corresponding ordered
monomials in $W^-\cup H\cup W^+$ form an $\F$-basis of $\Oq$.
\end{proof}

The case $H=\widetilde G$ proves
\cite[Conjecture~16.2]{Terwilliger2022OqACE}.

\begin{corollary}
\label{cor:alternating-polynomial-algebras}
The following four subalgebras of $\Oq$ are polynomial algebras in the
displayed generators:
\begin{equation*}
\begin{gathered}
 \F[W_0,W_{-1},W_{-2},\ldots],\qquad
 \F[W_1,W_2,W_3,\ldots],\\
 \F[G_1,G_2,G_3,\ldots],\qquad
 \F[\widetilde G_1,\widetilde G_2,\widetilde G_3,\ldots].
\end{gathered}
\end{equation*}
\end{corollary}

\begin{proof}
Applying the canonical central reduction to
\eqref{eq:Aq4} and \eqref{eq:Aq10} shows that the elements within each of
the four families commute pairwise.  By
Theorem~\ref{thm:twelve-PBW}, the monomials in any one of these families
are linearly independent.  Hence each family is algebraically independent
and generates the corresponding polynomial algebra.
\end{proof}

\begin{remark}
The proof of Lemma~\ref{lem:filtered-lifting} shows that the PBW bases in
Theorem~\ref{thm:twelve-PBW} are compatible with the word-length
filtration.  For every $m\geqslant 1$, the two generators $W_{-(m-1)}$ and
$W_m$ have filtration degree $2m-1$, while $G_m$ and
$\widetilde G_m$ have filtration degree $2m$.  Consequently,
\[
 \operatorname{Hilb}_{\gr\Oq}(t)
 =
 \prod_{m\geqslant 1}
 \frac{1}
 {(1-t^{2m-1})^2(1-t^{2m})}.
\]
\end{remark}

\section{Central specialisations of the alternating central extension}
\label{sec:central-specialisations}

We show that the twelve alternating PBW bases remain valid after arbitrary
scalar central specialisation of $\Aq$.

Put $z_0=1$.  By
\cite[Lemma~13.1 and Theorem~13.5]{Terwilliger2022OqACE}, there is an
algebra isomorphism
$\Phi:\Aq\longrightarrow\Oq\otimes\F[z_1,z_2,\ldots]$ such that, for
$n\in\N$,
\begin{align}
\Phi(\cW_{-n})
 &=
 \sum_{k=0}^{n}W_{k-n}\otimes z_k,
\label{eq:central-convolution-Wminus}\\
\Phi(\cW_{n+1})
 &=
 \sum_{k=0}^{n}W_{n+1-k}\otimes z_k,
\label{eq:central-convolution-Wplus}\\
\Phi(\cG_n)
 &=
 \sum_{k=0}^{n}G_{n-k}\otimes z_k,
\label{eq:central-convolution-G}\\
\Phi(\ctG_n)
 &=
 \sum_{k=0}^{n}\widetilde G_{n-k}\otimes z_k.
\label{eq:central-convolution-tG}
\end{align}
Here $W_{k-n}=W_{-(n-k)}$, while
$G_0=\widetilde G_0=-(q-q^{-1})[2]_q^2$ are scalars.

For $n\in\N$, let $\Zconv_n$ be the central element of $\Aq$
defined by
\begin{equation}
\label{eq:convolution-central-coordinate}
        \Phi(\Zconv_n)=1\otimes z_n.
\end{equation}
Thus $\Zconv_0=1$.  The elements
$\Zconv_1,\Zconv_2,\ldots$ are algebraically independent and generate $Z(\Aq)$; see
\cite[Definition~13.8 and Lemma~13.9]{Terwilliger2022OqACE}.  We refer to
them as the \textbf{convolution central coordinates}.

\subsection{Scalar central fibres}

Let
$\boldsymbol\zeta=(\zeta_1,\zeta_2,\ldots)
\in\F^{\N\setminus\{0\}}$, and put $\zeta_0=1$.  Evaluation at
$z_n=\zeta_n$ defines an algebra homomorphism
$\operatorname{ev}_{\boldsymbol\zeta}:
\F[z_1,z_2,\ldots]\longrightarrow\F$.  Set
\[
 \pi_{\boldsymbol\zeta}
 =
 \bigl(
 \operatorname{id}_{\Oq}\otimes
 \operatorname{ev}_{\boldsymbol\zeta}
 \bigr)\circ\Phi:
 \Aq\longrightarrow\Oq.
\]
Then $\pi_{\boldsymbol\zeta}(\Zconv_n)=\zeta_n$ for
$n\geqslant 1$.
Since the $\Zconv_n$ freely generate $Z(\Aq)$, every
$\F$-valued central character of $\Aq$ is obtained uniquely in this
way.

Define the corresponding central fibre by
\begin{equation}
\label{eq:central-fibre-definition}
 \Aqfib{\boldsymbol\zeta}
 =
 \Aq\big/
 \left\langle
        \Zconv_n-\zeta_n
        \mid n\geqslant 1
 \right\rangle.
\end{equation}
Under $\Phi$, the ideal in
\eqref{eq:central-fibre-definition} corresponds to
\[
 \Oq\otimes
 \left\langle
        z_n-\zeta_n
        \mid n\geqslant 1
 \right\rangle.
\]
Consequently, $\pi_{\boldsymbol\zeta}$ induces an algebra isomorphism
\begin{equation}
\label{eq:central-fibre-isomorphism}
        \Aqfib{\boldsymbol\zeta}
        \xrightarrow{\ \sim\ }
        \Oq.
\end{equation}
We use \eqref{eq:central-fibre-isomorphism} to identify the central fibre
with $\Oq$, equipped with its word-length filtration.

For $n\in\N$, let
$W_{-n}^{(\boldsymbol\zeta)}$,
$W_{n+1}^{(\boldsymbol\zeta)}$,
$G_n^{(\boldsymbol\zeta)}$, and
$\widetilde G_n^{(\boldsymbol\zeta)}$ denote the images under
$\pi_{\boldsymbol\zeta}$ of the corresponding calligraphic generators.
Equations
\eqref{eq:central-convolution-Wminus}--\eqref{eq:central-convolution-tG}
give
\begin{align}
W_{-n}^{(\boldsymbol\zeta)}
 &=
 \sum_{k=0}^{n}\zeta_kW_{k-n}
 =
 W_{-n}
 +
 \sum_{k=1}^{n}
 \zeta_kW_{-(n-k)},
\label{eq:special-Wminus}\\
W_{n+1}^{(\boldsymbol\zeta)}
 &=
 \sum_{k=0}^{n}\zeta_kW_{n+1-k}
 =
 W_{n+1}
 +
 \sum_{k=1}^{n}
 \zeta_kW_{n+1-k},
\label{eq:special-Wplus}\\
G_n^{(\boldsymbol\zeta)}
 &=
 \sum_{k=0}^{n}\zeta_kG_{n-k}
 =
 G_n
 +
 \sum_{k=1}^{n}
 \zeta_kG_{n-k},
\label{eq:special-G}\\
\widetilde G_n^{(\boldsymbol\zeta)}
 &=
 \sum_{k=0}^{n}\zeta_k\widetilde G_{n-k}
 =
 \widetilde G_n
 +
 \sum_{k=1}^{n}
 \zeta_k\widetilde G_{n-k}.
\label{eq:special-tG}
\end{align}
In the last two formulae, the summand with $k=n$ is a scalar multiple
of $G_0=\widetilde G_0=-(q-q^{-1})[2]_q^2$.

\begin{proposition}
\label{prop:special-leading}
For every $\boldsymbol\zeta$ and every $n\in\N$,
\begin{align}
\init(W_{-n}^{(\boldsymbol\zeta)})
 &=
 \init(W_{-n}),
\label{eq:special-leading-Wminus}\\
\init(W_{n+1}^{(\boldsymbol\zeta)})
 &=
 \init(W_{n+1}),
\label{eq:special-leading-Wplus}\\
\init(G_n^{(\boldsymbol\zeta)})
 &=
 \init(G_n),
\label{eq:special-leading-G}\\
\init(\widetilde G_n^{(\boldsymbol\zeta)})
 &=
 \init(\widetilde G_n).
\label{eq:special-leading-tG}
\end{align}
Consequently,
\[
 \deg_F W_{-n}^{(\boldsymbol\zeta)}
 =\deg_F W_{n+1}^{(\boldsymbol\zeta)}=2n+1,
 \qquad
 \deg_F G_n^{(\boldsymbol\zeta)}
 =\deg_F\widetilde G_n^{(\boldsymbol\zeta)}=2n.
\]
\end{proposition}

\begin{proof}
In each of
\eqref{eq:special-Wminus}--\eqref{eq:special-tG}, the term with $k=0$
is the corresponding canonical alternating generator.

For $1\leqslant k\leqslant n$, Theorem~\ref{thm:all-leading} gives
\begin{align*}
 \deg_F W_{-(n-k)}
 &=\deg_F W_{n+1-k}=2(n-k)+1<2n+1,\\
 \deg_F G_{n-k}
 &=\deg_F\widetilde G_{n-k}=2(n-k)<2n.
\end{align*}
Thus every term with $k\geqslant 1$ has filtration degree strictly smaller
than the term with $k=0$.  The initial forms are therefore unchanged.
The degree statements follow from Theorem~\ref{thm:all-leading}.
\end{proof}

Set
\begin{equation*}
\begin{aligned}
 W^-_{\boldsymbol\zeta}
 &=\{W_{-i}^{(\boldsymbol\zeta)}\mid i\in\N\},&
 W^+_{\boldsymbol\zeta}
 &=\{W_{i+1}^{(\boldsymbol\zeta)}\mid i\in\N\},\\
 G_{\boldsymbol\zeta}
 &=\{G_{i+1}^{(\boldsymbol\zeta)}\mid i\in\N\},&
 \widetilde G_{\boldsymbol\zeta}
 &=\{\widetilde G_{i+1}^{(\boldsymbol\zeta)}\mid i\in\N\}.
\end{aligned}
\end{equation*}

\begin{theorem}
\label{thm:central-fibre-PBW}
Let
$H_{\boldsymbol\zeta}\in
\{G_{\boldsymbol\zeta},\widetilde G_{\boldsymbol\zeta}\}$, and let
$\prec$ be a total order on
$W^-_{\boldsymbol\zeta}\cup H_{\boldsymbol\zeta}
\cup W^+_{\boldsymbol\zeta}$.
Suppose that $\prec$ satisfies one of the six block conditions in
\eqref{eq:six-Oq-orders}, after replacing $W^-$, $H$, and $W^+$
by $W^-_{\boldsymbol\zeta}$, $H_{\boldsymbol\zeta}$, and
$W^+_{\boldsymbol\zeta}$, respectively.  Then the $\prec$-ordered monomials in the specialised
alternating generators form an $\F$-basis of
$\Aqfib{\boldsymbol\zeta}$.
\end{theorem}

\begin{proof}
We use \eqref{eq:central-fibre-isomorphism} to regard the specialised
generators as elements of $\Oq$.

For every
$X^{(\boldsymbol\zeta)}\in
W^-_{\boldsymbol\zeta}\cup H_{\boldsymbol\zeta}
\cup W^+_{\boldsymbol\zeta}$, let $X^{\sh}$ be the corresponding
alternating word.  By Proposition~\ref{prop:special-leading} and
Theorem~\ref{thm:all-leading},
$\init(X^{(\boldsymbol\zeta)})=c_XX^{\sh}$ for some
$c_X\in\F^\times$.

Transfer $\prec$ to the corresponding alternating words.  By
Theorem~\ref{thm:qshuffle-twelve}, their $\prec$-ordered monomials form an $\F$-basis of $\Uplus\simeq\gr\Oq$.
Lemma~\ref{lem:filtered-lifting} therefore shows that the corresponding
ordered monomials in the specialised alternating generators form an
$\F$-basis of $\Oq$, and hence of
$\Aqfib{\boldsymbol\zeta}$.
\end{proof}

\subsection{Comparison of central coordinates}

The preceding argument uses only the convolution coordinates
$\{\Zconv_n\}_{n\geqslant 1}$.  We finish by comparing them with the
central coordinates used in Section~\ref{sec:BB}.

Put $h=[2]_q=q+q^{-1}$.  Let
$\Zvee_0,\Zvee_1,\Zvee_2,\ldots$ be the normalised central elements of
\cite[Definitions~8.3 and~8.4]{Terwilliger2022OqACE}; in particular,
$\Zvee_0=1$.  By
\cite[Lemma~13.10(i)]{Terwilliger2022OqACE}, for every
$n\geqslant 1$ there is a polynomial
$P_n\in\F[X_1,\ldots,X_{n-1}]$ with zero constant term such that
\begin{equation}
\label{eq:Zvee-Zconv-triangular}
 \Zvee_n
 =
 h^n(q^n+q^{-n})\Zconv_n
 +
 P_n(\Zconv_1,\ldots,\Zconv_{n-1}).
\end{equation}
The polynomial $P_n$ has weighted degree at most $n$ when
$\deg X_k=k$.  In particular,
\begin{equation}
\label{eq:Zvee1-Zconv1}
        \Zvee_1=h^2\Zconv_1.
\end{equation}

The coefficient $h^n(q^n+q^{-n})$ is non-zero by Lemma~\ref{lem:nonvanishing-scalars}.  Moreover,
\cite[Lemma~13.10(ii)]{Terwilliger2022OqACE} gives the inverse triangular
system: $\Zconv_n$ is a polynomial in
$\Zvee_1,\ldots,\Zvee_n$, with coefficient $h^{-n}(q^n+q^{-n})^{-1}$ at $\Zvee_n$.
Consequently, prescribing the scalar values of the
$\Zvee_n$ is equivalent to prescribing those of the $\Zconv_n$.

Assume now that $\operatorname{char}\F\neq 2$.  Let $\Delta_n$
denote the central elements used by Baseilhac and Belliard.  Combining
\cite[Definitions~8.3 and~8.4]{Terwilliger2022OqACE} with
\cite[Remark~8.17]{Terwilliger2021ACE} gives
\begin{equation}
\label{eq:Delta-Zvee}
 \Delta_n
 =
 -2
 \frac{(q-q^{-1})h^{2-n}}
      {q^n+q^{-n}}
 \Zvee_n
 \qquad(n\geqslant 1).
\end{equation}
Every scalar coefficient in \eqref{eq:Delta-Zvee} is non-zero.  Hence, in
characteristic different from $2$, prescribing the values of the
$\Delta_n$, the $\Zvee_n$, or the $\Zconv_n$ gives equivalent
systems of scalar central specialisations.

No restriction on $\operatorname{char}\F$ is required for
Theorem~\ref{thm:central-fibre-PBW}; the characteristic assumption is used
only in the comparison with the Baseilhac--Belliard coordinates
$\Delta_n$.

\section{The Baseilhac--Belliard
\texorpdfstring{$WG$}{WG}-basis}
\label{sec:BB}

Throughout this section, we assume that $\operatorname{char}\F\neq 2$.
This assumption is used only to pass between the central coordinates
$\Delta_n$ of Baseilhac and Belliard and the central coordinates of
Section~\ref{sec:central-specialisations}.

For $\rho\in\F^\times$, let $\Aq(\rho)$ denote the current algebra
obtained from Definition~\ref{def:Aq} by replacing $\rho_0$ in
\eqref{eq:Aq2} and \eqref{eq:Aq3} by $\rho$.  Let
$\Delta_1^{(\rho)},\Delta_2^{(\rho)},\ldots$ be the corresponding
central elements of Baseilhac and Belliard.  For
$\boldsymbol\delta=(\delta_1,\delta_2,\ldots)
\in\F^{\N\setminus\{0\}}$, set
\begin{equation}
\label{eq:BB-quotient-rho}
 \BBqquot{\rho}{\boldsymbol\delta}
 =
 \Aq(\rho)\big/
 \left\langle
        \Delta_n^{(\rho)}-2\delta_n
        \mid n\geqslant 1
 \right\rangle.
\end{equation}
These are the quotients of
\cite[Definition~3.1]{BaseilhacBelliard2017}.  For $\rho=\rho_0$, we identify $\Aq(\rho_0)=\Aq$ and omit the
superscript $(\rho_0)$ from the central elements.

Whenever a quotient
$\BBqquot{\rho}{\boldsymbol\delta}$ is fixed, the symbols
$W^-$, $W^+$, $G$, and $\widetilde G$ denote the images in the quotient of the corresponding four current
families.

\subsection{The normalised parameter}

For $n\geqslant 1$, define
\begin{equation}
\label{eq:BB-Zvee-values}
 \eta_n
 =
 -\frac{(q^n+q^{-n})[2]_q^{\,n-2}}
        {q-q^{-1}}\,
 \delta_n.
\end{equation}
By \eqref{eq:Delta-Zvee}, the relation $\Delta_n=2\delta_n$ is
equivalent to $\Zvee_n=\eta_n$.

\begin{lemma}
\label{lem:BB-central-fibre}
There is a unique sequence
$\boldsymbol\zeta=(\zeta_1,\zeta_2,\ldots)
\in\F^{\N\setminus\{0\}}$ such that
\begin{equation}
\label{eq:BB-central-ideal-equality}
 \left\langle
        \Delta_n-2\delta_n
        \mid n\geqslant 1
 \right\rangle
 =
 \left\langle
        \Zvee_n-\eta_n
        \mid n\geqslant 1
 \right\rangle
 =
 \left\langle
        \Zconv_n-\zeta_n
        \mid n\geqslant 1
 \right\rangle.
\end{equation}
Consequently, there is an algebra isomorphism
\begin{equation}
\label{eq:BB-normalised-fibre-isomorphism}
        \BBqquot{\rho_0}{\boldsymbol\delta}
        \xrightarrow{\ \sim\ }
        \Aqfib{\boldsymbol\zeta}.
\end{equation}
Under this isomorphism, the images of the current generators are the
specialised alternating generators
$W_{-n}^{(\boldsymbol\zeta)}$,
$W_{n+1}^{(\boldsymbol\zeta)}$,
$G_{n+1}^{(\boldsymbol\zeta)}$, and
$\widetilde G_{n+1}^{(\boldsymbol\zeta)}$, for $n\in\N$.
\end{lemma}

\begin{proof}
By \eqref{eq:Delta-Zvee},
\[
 \Delta_n-2\delta_n
 =
 -2\frac{(q-q^{-1})[2]_q^{\,2-n}}
          {q^n+q^{-n}}
 \bigl(\Zvee_n-\eta_n\bigr).
\]
The scalar multiplying
$\Zvee_n-\eta_n$ is non-zero by
Lemma~\ref{lem:nonvanishing-scalars} and the assumption
$\operatorname{char}\F\neq 2$.  This proves the first equality in
\eqref{eq:BB-central-ideal-equality}.

For $n\geqslant 1$, put $a_n=[2]_q^n(q^n+q^{-n})$.  By
\eqref{eq:Zvee-Zconv-triangular},
\[
 \Zvee_n
 =
 a_n\Zconv_n
 +
 P_n(\Zconv_1,\ldots,\Zconv_{n-1}),
\]
where $a_n\neq 0$.  Define $\zeta_n$ recursively by
\begin{equation}
\label{eq:BB-zeta-recursion}
 \zeta_n
 =
 a_n^{-1}
 \bigl(
 \eta_n-P_n(\zeta_1,\ldots,\zeta_{n-1})
 \bigr).
\end{equation}
For $n=1$, the polynomial $P_1$ is zero.

We show by induction on $n$ that
\begin{equation}
\label{eq:BB-truncated-ideal-equality}
 \left\langle
        \Zvee_k-\eta_k
        \mid 1\leqslant k\leqslant n
 \right\rangle
 =
 \left\langle
        \Zconv_k-\zeta_k
        \mid 1\leqslant k\leqslant n
 \right\rangle.
\end{equation}
Indeed,
\begin{align*}
\Zvee_n-\eta_n
={}&
a_n(\Zconv_n-\zeta_n)+
P_n(\Zconv_1,\ldots,\Zconv_{n-1})
-
P_n(\zeta_1,\ldots,\zeta_{n-1}).
\end{align*}
The second line belongs to the ideal generated by
$\Zconv_k-\zeta_k$ for $1\leqslant k<n$.
Since $a_n$ is invertible, the induction hypothesis gives
\eqref{eq:BB-truncated-ideal-equality}.  Taking the union over $n$
proves the second equality in
\eqref{eq:BB-central-ideal-equality}.  The same recursion also proves the
uniqueness of $\boldsymbol\zeta$.

The isomorphism
\eqref{eq:BB-normalised-fibre-isomorphism} now follows from the definition
of $\Aqfib{\boldsymbol\zeta}$.  The assertion concerning the current
generators follows from
\eqref{eq:special-Wminus}--\eqref{eq:special-tG}.
\end{proof}

\begin{remark}
\label{rem:BB-low-degree}
For $n=1$, equations
\eqref{eq:BB-Zvee-values} and
\eqref{eq:BB-zeta-recursion} give
\[
        \eta_1
        =
        -\frac{\delta_1}{q-q^{-1}},
        \qquad
        \zeta_1
        =
        -\frac{\delta_1}
        {(q-q^{-1})[2]_q^2}.
\]
Since $G_0=\widetilde G_0=-(q-q^{-1})[2]_q^2$, we obtain
\begin{align}
G_1^{(\boldsymbol\zeta)}
 &=
 G_1+\zeta_1G_0
 =
 \qcomm{W_1}{W_0}+\delta_1,
\label{eq:BB-low-check}\\
\widetilde G_1^{(\boldsymbol\zeta)}
 &=
 \widetilde G_1+\zeta_1\widetilde G_0
 =
 \qcomm{W_0}{W_1}+\delta_1.
\end{align}
Here we used
\eqref{eq:tG1-Bdelta} and
\eqref{eq:G1-low-check}.  These are precisely the formulae in
\cite[equation~(3.5)]{BaseilhacBelliard2017}.
\end{remark}

For later reference, consider the monomials
\begin{equation}
\label{eq:BB-monomials}
 W_{-k_1}^{a_1}\cdots W_{-k_N}^{a_N}
 G_{p_1+1}^{b_1}\cdots G_{p_P+1}^{b_P}
 W_{\ell_M+1}^{c_M}\cdots W_{\ell_1+1}^{c_1},
\end{equation}
where $N,P,M\in\N$,
\[
 0\leqslant k_1<\cdots<k_N,\qquad
 0\leqslant p_1<\cdots<p_P,\qquad
 0\leqslant \ell_1<\cdots<\ell_M,
\]
and $a_i,b_j,c_t\in\Z_{>0}$.
Empty products are interpreted as $1$.

\begin{theorem}
\label{thm:BB-normalised}
Let $H\in\{G,\widetilde G\}$, and let $\prec$ be a total order on
$W^-\cup H\cup W^+$ satisfying one of the six block conditions in
\eqref{eq:six-Oq-orders}.  Then the $\prec$-ordered monomials form an
$\F$-basis of
$\BBqquot{\rho_0}{\boldsymbol\delta}$.

In particular, the monomials in
\eqref{eq:BB-monomials} form an $\F$-basis of
$\BBqquot{\rho_0}{\boldsymbol\delta}$.
\end{theorem}

\begin{proof}
By Lemma~\ref{lem:BB-central-fibre}, the quotient is the central fibre
$\Aqfib{\boldsymbol\zeta}$, and its current generators are the
specialised alternating generators.  The first assertion is therefore
Theorem~\ref{thm:central-fibre-PBW}.

For the final assertion, take the block order $W^-<G<W^+$, order the $W^-$- and $G$-blocks by increasing index, and order the
$W^+$-block by decreasing index.  The resulting ordered monomials are
exactly those in \eqref{eq:BB-monomials}.
\end{proof}

\subsection{Arbitrary non-zero parameter}

The quantum-determinant generating series in
\cite[Proposition~2.1]{BaseilhacBelliard2017} is a linear combination of
terms of the following three types:
\[
        \cG+\ctG,\qquad
        \cW\cW,\qquad
        \rho^{-1}\cG\ctG.
\]
Consequently, each coefficient $\Delta_n^{(\rho)}$ has scaling degree
$2$ under the assignments $\deg_s\cW=1$,
$\deg_s\cG=\deg_s\ctG=2$, and $\deg_s\rho=2$.

\begin{lemma}
\label{lem:rho-scaling}
Let $\rho\in\F^\times$, and let $\F'/\F$ be a field extension
containing $s\in\F'^\times$ such that $s^2\rho_0=\rho$.  There is
an $\F'$-algebra isomorphism
\[
 \Theta_s:
 \Aq(\rho)\otimes_\F\F'
 \longrightarrow
 \Aq(\rho_0)\otimes_\F\F'
\]
determined, for $n\in\N$, by
\begin{align*}
\cW_{-n}^{(\rho)}
 &\longmapsto
 s\,\cW_{-n}^{(\rho_0)},
&
\cW_{n+1}^{(\rho)}
 &\longmapsto
 s\,\cW_{n+1}^{(\rho_0)},\\
\cG_{n+1}^{(\rho)}
 &\longmapsto
 s^2\cG_{n+1}^{(\rho_0)},
&
\ctG_{n+1}^{(\rho)}
 &\longmapsto
 s^2\ctG_{n+1}^{(\rho_0)}.
\end{align*}
Moreover,
\begin{equation}
\label{eq:Delta-scaling}
        \Theta_s(\Delta_n^{(\rho)})
        =
        s^2\Delta_n^{(\rho_0)}
        \qquad(n\geqslant 1).
\end{equation}
Hence $\Theta_s$ induces an algebra isomorphism
\begin{equation}
\label{eq:BB-quotient-scaling}
\begin{aligned}
\overline\Theta_s:\quad
\BBqquot{\rho}{\boldsymbol\delta}\otimes_\F\F'
&\longrightarrow
\bigl(\Aq(\rho_0)\otimes_\F\F'\bigr)
\Big/
\left\langle
 \Delta_n^{(\rho_0)}-2s^{-2}\delta_n
 \mid n\geqslant 1
\right\rangle.
\end{aligned}
\end{equation}
\end{lemma}

\begin{proof}
Under the proposed assignment, both sides of
\eqref{eq:Aq1} and
\eqref{eq:Aq4}--\eqref{eq:Aq11} are multiplied by the same power of
$s$.  In \eqref{eq:Aq2} and \eqref{eq:Aq3}, the left-hand side is
multiplied by $s^3$, while each parameter-dependent term is multiplied
by $\rho s=s^3\rho_0$.  Thus all defining relations are preserved.  The same assignment with
$s^{-1}$ gives the inverse map.

For the central elements, the terms of type $\cG+\ctG$ and
$\cW\cW$ are multiplied by $s^2$.  A term of type $\rho^{-1}\cG\ctG$ is also multiplied by $s^2$,
since $\rho^{-1}s^4=(s^2\rho_0)^{-1}s^4=s^2\rho_0^{-1}$.
This proves \eqref{eq:Delta-scaling}.

Finally,
\[
 \Theta_s
 \bigl(
 \Delta_n^{(\rho)}-2\delta_n
 \bigr)
 =
 s^2
 \bigl(
 \Delta_n^{(\rho_0)}-2s^{-2}\delta_n
 \bigr).
\]
Since $s^2$ is invertible, $\Theta_s$ identifies the two defining
ideals and therefore induces
\eqref{eq:BB-quotient-scaling}.
\end{proof}

\begin{theorem}
\label{thm:BB-general}
Let $\rho\in\F^\times$, and let
\[
        \boldsymbol\delta
        =
        (\delta_1,\delta_2,\ldots)
        \in\F^{\N\setminus\{0\}}.
\]
Let $H\in\{G,\widetilde G\}$, and let $\prec$ be a total order on
$W^-\cup H\cup W^+$ satisfying one of the six block conditions in
\eqref{eq:six-Oq-orders}.  Then the $\prec$-ordered monomials form an
$\F$-basis of
$\BBqquot{\rho}{\boldsymbol\delta}$.

In particular, the monomials in
\eqref{eq:BB-monomials} form an $\F$-basis of
$\BBqquot{\rho}{\boldsymbol\delta}$.
\end{theorem}

\begin{proof}
By Lemma~\ref{lem:nonvanishing-scalars},
$\rho_0\neq 0$.  Choose a field extension $\F'/\F$ and an element
$s\in\F'^\times$ such that $s^2\rho_0=\rho$.
The parameter $q$ is still not a root of unity in $\F'$, and
$\operatorname{char}\F'=\operatorname{char}\F$.

By Lemma~\ref{lem:rho-scaling},
$\BBqquot{\rho}{\boldsymbol\delta}\otimes_\F\F'$ is isomorphic to the
normalised quotient over $\F'$ with central values
$s^{-2}\boldsymbol\delta=(s^{-2}\delta_1,s^{-2}\delta_2,\ldots)$.
By Theorem~\ref{thm:BB-normalised}, the chosen ordered monomials form a
basis of this normalised quotient.  Under
$\overline\Theta_s$, every ordered monomial in the original quotient is
sent to a non-zero scalar multiple of the corresponding ordered monomial
in the normalised quotient.  Hence the chosen monomials form an
$\F'$-basis after extending scalars.

Let $M=\BBqquot{\rho}{\boldsymbol\delta}$, and let $V$ be the free
$\F$-vector space with basis indexed by the chosen ordered monomials.
There is a natural linear map $\theta:V\longrightarrow M$ sending each formal basis vector to the corresponding monomial.  The preceding paragraph shows that
$\theta\otimes_{\F}\operatorname{id}_{\F'}$ is an isomorphism.  Since $\F'/\F$ is faithfully flat,
\[
        \ker\theta=0,
        \qquad
        \operatorname{coker}\theta=0.
\]
Thus $\theta$ is an isomorphism, and the ordered monomials form an
$\F$-basis of $M$.

The final assertion is obtained by taking $H=G$, increasing index order
in the $W^-$- and $G$-blocks, and decreasing index order in the
$W^+$-block.
\end{proof}

The particular basis in \eqref{eq:BB-monomials} is the $WG$-basis
proposed in \cite[Conjecture~1]{BaseilhacBelliard2017}.  Thus that
conjecture holds.

\section{The negative alternating centraliser}
\label{sec:negative-centraliser}

We work in the $q$-shuffle realisation
$U\subseteq(\Vsh,\star)$ fixed in Section~\ref{sec:qshuffle}.  In this
section only, we omit the superscript $\sh$ from the alternating words and
write
\begin{equation*}
 W_{-n}=x(yx)^n,
 \qquad
 G_n=(yx)^n
 \qquad(n\in\N).
\end{equation*}
All products in this section are $q$-shuffle products.  Put
\begin{equation}
\label{eq:negative-shuffle-algebra}
 \mathcal N
 =
 \F[W_0,W_{-1},W_{-2},\ldots]
 \subseteq U.
\end{equation}
By \cite[Lemma~5.12(i)]{Terwilliger2019Alternating}, the negative
alternating words commute pairwise.  By
Theorem~\ref{thm:qshuffle-twelve}, their monomials are linearly
independent.  Hence $\mathcal N$ is a polynomial algebra.

\subsection{A left-deletion operator}

For $a,b\in\N$, let $\Vsh_{a,b}$ denote the span of the words containing
$a$ copies of $x$ and $b$ copies of $y$.  The $q$-shuffle product makes
$\Vsh$ an $\N^2$-graded algebra, and $U$ is a graded subalgebra.  If
$u\in\Vsh_{a,b}$ is homogeneous, put
\begin{equation*}
        d(u)=a-b.
\end{equation*}
Set $t=q^2$.  For $m\geqslant 1$, write
\begin{equation}
\label{eq:t-integer}
 \langle m\rangle_t
 =1+t+\cdots+t^{m-1}
 =\frac{t^m-1}{t-1}.
\end{equation}
Since $q$ is not a root of unity, neither is $t$.  Consequently,
\begin{equation}
\label{eq:t-scalars-nonzero}
        t^m-1\neq 0,
        \qquad
        \langle m\rangle_t\neq 0
        \qquad(m\geqslant 1).
\end{equation}

Define an $\F$-linear map $D:\Vsh\longrightarrow\Vsh$ by
\begin{equation}
\label{eq:left-deletion-definition}
        D(1)=0,
        \qquad
        D(xw)=w,
        \qquad
        D(yw)=0
\end{equation}
for every word $w$.  Thus $D$ deletes an initial $x$ and annihilates a
word that does not begin with $x$.

\begin{lemma}
\label{lem:left-deletion-Leibniz}
Let $u\in\Vsh$ be homogeneous and let $v\in\Vsh$.  Then
\begin{equation}
\label{eq:left-deletion-Leibniz}
 D(u\star v)
 =
 D(u)\star v+t^{d(u)}u\star D(v).
\end{equation}
In particular, $D(U)\subseteq U$.
\end{lemma}

\begin{proof}
It suffices to prove \eqref{eq:left-deletion-Leibniz} when $u$ and $v$
are words.  In the shuffle expansion of $u\star v$, a word that begins
with $x$ has its first letter either from $u$ or from $v$.  The terms of
the first type give $D(u)\star v$.  For a term of the second type, the
initial $x$ from $v$ crosses every letter of $u$, and therefore contributes
the scalar
\begin{equation*}
 q^{\langle\deg u,x\rangle}
 =q^{2d(u)}
 =t^{d(u)}.
\end{equation*}
These terms give $t^{d(u)}u\star D(v)$, proving the formula.

Since $D(x)=1$ and $D(y)=0$, the final assertion follows from
\eqref{eq:left-deletion-Leibniz} by induction on the length of a product
in the generators $x,y$.
\end{proof}

Define endomorphisms $L,R$ of $\Vsh$ by
\begin{equation*}
        L(u)=x\star u,
        \qquad
        R(u)=u\star x,
\end{equation*}
and, for $\alpha\in\F$, put
\begin{equation}
\label{eq:C-alpha-definition}
        C_\alpha=R-\alpha L.
\end{equation}
The operators $L$ and $R$ commute, so the operators $C_\alpha$ commute
pairwise.  Equation~\eqref{eq:left-deletion-Leibniz} gives, for every
homogeneous $u\in\Vsh$,
\begin{equation}
\label{eq:D-C-alpha}
 D C_\alpha(u)
 =
 C_{\alpha t}D(u)
 +\bigl(t^{d(u)}-\alpha\bigr)u.
\end{equation}
Notice that
\begin{equation}
\label{eq:C1-commutator}
        C_1(u)=u\star x-x\star u=[u,x].
\end{equation}

Order the words lexicographically, with $x<y$.  For a non-zero element
$z\in\Vsh$, let $\operatorname{lw}(z)$ denote the least word that occurs
in $z$ with non-zero coefficient.

\begin{lemma}
\label{lem:x-centraliser-deletion}
Let $0\neq z\in\Vsh_{a,b}$ satisfy $[z,x]=0$, and put $r=a-b$.  Then
\begin{equation}
\label{eq:x-centraliser-deletion}
        r\geqslant 0,
        \qquad
        D^{r+1}z=0,
        \qquad
        D^rz\neq 0.
\end{equation}
Consequently, for $r,b\in\N$, the map $D^r$ is injective on
$\Cent_{\Vsh}(x)\cap\Vsh_{r+b,b}$.
\end{lemma}

\begin{proof}
Let $w=\operatorname{lw}(z)$, and write
$w=x^{a_0}w'$, where $w'$ is empty or begins with $y$.  Let
$c_w\in\F^\times$ be the coefficient of $w$ in $z$.

Consider the coefficient of the word $xw$ in $[z,x]$.  In
$x\star w$, the extra $x$ can be inserted at any place in the initial
run of $x$'s, and the sum of the resulting coefficients is
$\langle a_0+1\rangle_t$.  In $w\star x$, the corresponding sum is
$t^{r-a_0}\langle a_0+1\rangle_t$.

Suppose that a word $v\neq w$ can also produce $xw$ by the insertion of
one copy of $x$.  Then $v$ is obtained from $xw$ by deleting an $x$ to
the right of the initial run.  It therefore begins with at least
$a_0+1$ copies of $x$, and hence $v<w$.  By the definition of $w$, no
such word occurs in $z$.  It follows that the coefficient of $xw$ in
$[z,x]$ is
\begin{equation*}
 c_w\langle a_0+1\rangle_t
 \bigl(t^{r-a_0}-1\bigr).
\end{equation*}
This coefficient is zero.  By \eqref{eq:t-scalars-nonzero}, we obtain
$t^{r-a_0}=1$.  Since $t$ is not a root of unity,
$r=a_0\geqslant 0$.

If $b>0$, no word occurring in $z$ can begin with more than $r$ copies
of $x$, since such a word would be smaller than $w$.  If $b=0$, the
bidegree contains only the word $x^a$.  Hence $D^{r+1}z=0$.  On the
other hand, the coefficient of the suffix obtained from $w$ by deleting
its initial $r$ copies of $x$ remains $c_w$ in $D^rz$.  Thus
$D^rz\neq 0$.  The injectivity assertion follows immediately.
\end{proof}

For $r\in\N$, define
\begin{equation}
\label{eq:S-r-definition}
        S_r=C_tC_{t^2}\cdots C_{t^r},
        \qquad
        S_0=\operatorname{id}_{\Vsh}.
\end{equation}
Also put
\begin{equation*}
 \langle r\rangle_t!
 =\prod_{j=1}^{r}\langle j\rangle_t,
 \qquad
 \langle 0\rangle_t!=1.
\end{equation*}

\begin{lemma}
\label{lem:x-centraliser-reconstruction}
Let $0\neq z\in\Vsh$ be homogeneous, satisfy $[z,x]=0$, and put
$r=d(z)$.  Then $r\geqslant 0$.  Put $v=D^rz$.  One has $Dv=0$ and
\begin{equation}
\label{eq:x-centraliser-reconstruction}
 z
 =
 \gamma_rS_rv,
 \qquad
 \gamma_r
 =
 \frac{(-1)^r}
 {\langle r\rangle_t!\displaystyle\prod_{j=1}^{r}(t^j-1)}.
\end{equation}
The empty products in \eqref{eq:x-centraliser-reconstruction} are
interpreted as $1$.
\end{lemma}

\begin{proof}
The assertions $r\geqslant 0$ and $Dv=0$ follow from
Lemma~\ref{lem:x-centraliser-deletion}.  Since $C_1z=0$, equation
\eqref{eq:D-C-alpha} gives
\begin{equation*}
        C_tDz=-(t^r-1)z.
\end{equation*}
More generally, for $0\leqslant k<r$,
\begin{equation}
\label{eq:iterated-C-D}
 C_{t^{k+1}}D^{k+1}z
 =
 -\langle k+1\rangle_t
 \bigl(t^{r-k}-1\bigr)D^kz.
\end{equation}
We prove this by induction on $k$.  The case $k=0$ was just obtained.
For the induction step, apply $D$ to the preceding instance of
\eqref{eq:iterated-C-D} and use \eqref{eq:D-C-alpha}.  The required
coefficient identity is
\begin{equation*}
 \langle k\rangle_t\bigl(t^{r-k+1}-1\bigr)
 +t^{r-k}-t^k
 =
 \langle k+1\rangle_t\bigl(t^{r-k}-1\bigr).
\end{equation*}

Solving \eqref{eq:iterated-C-D} successively for
$D^kz$, for $k=0,1,\ldots,r-1$, and using the commutativity of the
operators $C_{t^j}$ gives \eqref{eq:x-centraliser-reconstruction}.
All denominators are non-zero by \eqref{eq:t-scalars-nonzero}.
\end{proof}

\subsection{Reduction to the word \texorpdfstring{$G_1$}{G1}}

Put
\begin{equation}
\label{eq:W-G-for-reduction}
        W=W_{-1}=xyx,
        \qquad
        G=G_1=yx.
\end{equation}
Then
\begin{equation}
\label{eq:deletion-W-G}
        DW=G,
        \qquad
        DG=0,
        \qquad
        [W,x]=0.
\end{equation}
The last equality follows from the pairwise commutativity of the negative
alternating words.

For $a\in\Z$ and $n\in\N$, define
\begin{equation}
\label{eq:T-a-n-definition}
 T_{a,n}
 =C_{t^a}C_{t^{a+1}}\cdots C_{t^{a+n-1}},
 \qquad
 T_{a,0}=\operatorname{id}_{\Vsh}.
\end{equation}

\begin{lemma}
\label{lem:deletion-operator-products}
The following statements hold.
\begin{enumerate}[label=\textup{(\roman*)}]
\item If $r\in\N$ and $u$ is homogeneous with $d(u)=1$, then
\begin{equation}
\label{eq:D-S-r}
        DS_ru=T_{2,r}Du.
\end{equation}

\item If $a\in\Z$, $n\geqslant 1$, and $w$ is homogeneous with
$d(w)=0$ and $Dw=0$, then
\begin{equation}
\label{eq:D-T-a-n}
 DT_{a,n}w
 =
 -(t^a-1)\langle n\rangle_tT_{a+1,n-1}w.
\end{equation}
Consequently, for $r\in\N$,
\begin{equation}
\label{eq:iterated-D-T}
 D^rT_{2,r}w
 =
 (-1)^r\langle r\rangle_t!
 \prod_{j=2}^{r+1}(t^j-1)w.
\end{equation}
\end{enumerate}
\end{lemma}

\begin{proof}
For (i), use the commutativity of the operators $C_\alpha$ to write the
factors in $S_r$ in decreasing order.  When
\eqref{eq:D-C-alpha} is applied to the factor $C_{t^j}$, its argument
has difference $j$, so the correction term vanishes.  Moving $D$ through
all the factors gives \eqref{eq:D-S-r}.

We prove \eqref{eq:D-T-a-n} by induction on $n$.  For $n=1$, it follows
directly from \eqref{eq:D-C-alpha}.  For $n\geqslant 2$, write
$T_{a,n}=C_{t^{a+n-1}}T_{a,n-1}$ and apply
\eqref{eq:D-C-alpha}.  After using the induction hypothesis, the result
follows from the operator identity
\begin{align*}
&-(t^a-1)\langle n-1\rangle_tC_{t^{a+n}}
 +\bigl(t^{n-1}-t^{a+n-1}\bigr)C_{t^a}\\
&\hspace{35mm}
 =-(t^a-1)\langle n\rangle_tC_{t^{a+n-1}},
\end{align*}
which is verified by comparing the coefficients of $R$ and $L$.
Iterating \eqref{eq:D-T-a-n}, beginning with $a=2$ and $n=r$, gives
\eqref{eq:iterated-D-T}.
\end{proof}

\begin{proposition}
\label{prop:two-generator-reduction}
Let $0\neq z\in U$ be homogeneous of difference $r$ and satisfy
$[z,x]=0$.  Put $v=D^rz$.  Then
\begin{equation}
\label{eq:two-generator-reduction}
        [z,W]=0
        \quad\Longleftrightarrow\quad
        [v,G]=0.
\end{equation}
\end{proposition}

\begin{proof}
By Lemma~\ref{lem:x-centraliser-reconstruction}, one has
$Dv=0$ and $z=\gamma_rS_rv$.  Since $[W,x]=0$, the map
$u\mapsto[u,W]$ commutes with $L$, $R$, and hence with $S_r$.  Therefore
\begin{equation}
\label{eq:z-W-S-r}
        [z,W]=\gamma_rS_r[v,W].
\end{equation}
By \eqref{eq:left-deletion-Leibniz} and
\eqref{eq:deletion-W-G},
\begin{equation}
\label{eq:deleting-v-commutators}
        D[v,W]=[v,G],
        \qquad
        D[v,G]=0.
\end{equation}
The element $[v,W]$ has difference $1$.  Applying
Lemma~\ref{lem:deletion-operator-products} to
\eqref{eq:z-W-S-r}, and then using
\eqref{eq:x-centraliser-reconstruction}, gives
\begin{equation}
\label{eq:key-two-generator-identity}
        D^{r+1}[z,W]
        =
        \langle r+1\rangle_t[v,G].
\end{equation}
Indeed, the scalar that occurs on the right-hand side is
\begin{equation*}
 \gamma_r(-1)^r\langle r\rangle_t!
 \prod_{j=2}^{r+1}(t^j-1)
 =
 \frac{t^{r+1}-1}{t-1}
 =
 \langle r+1\rangle_t.
\end{equation*}

If $[z,W]=0$, then \eqref{eq:key-two-generator-identity} and
\eqref{eq:t-scalars-nonzero} give $[v,G]=0$.  Conversely, suppose that
$[v,G]=0$.  Then \eqref{eq:key-two-generator-identity} gives
$D^{r+1}[z,W]=0$.  Moreover, the Jacobi identity and
$[z,x]=[W,x]=0$ show that $[z,W]$ commutes with $x$; it has
difference $r+1$.  Lemma~\ref{lem:x-centraliser-deletion} therefore
implies $[z,W]=0$.
\end{proof}

\subsection{An auxiliary imaginary centraliser in \texorpdfstring{$U$}{U}}

Retain the notation
\begin{equation*}
 A_i=E_{i\delta+\alpha_0},
 \qquad
 B_i=E_{i\delta+\alpha_1}
 \qquad(i\in\N),
\end{equation*}
and put $J_m=E_{m\delta}$ for $m\geqslant 1$.  By
\eqref{eq:Damiani-real} and \eqref{eq:Damiani-cross-II},
\begin{align}
[J_1,A_i]&=[2]_qA_{i+1},
&
[J_1,B_i]&=-[2]_qB_{i+1},
&
[J_1,J_m]&=0.
\label{eq:J1-root-shifts}
\end{align}
The adjacent cases of
\eqref{eq:Damiani-cross-AA-odd} and
\eqref{eq:Damiani-cross-BB-odd} give
\begin{equation}
\label{eq:adjacent-Damiani-real}
 A_{i+1}\star A_i=q^{-2}A_i\star A_{i+1},
 \qquad
 B_i\star B_{i+1}=q^{-2}B_{i+1}\star B_i.
\end{equation}

\begin{proposition}
\label{prop:J1-centraliser-shuffle}
One has
\begin{equation}
\label{eq:J1-centraliser-shuffle}
        \Cent_U(J_1)
        =
        \F[J_1,J_2,J_3,\ldots].
\end{equation}
\end{proposition}

\begin{proof}
The inclusion from right to left follows from
\eqref{eq:Damiani-cross-II}.  For the reverse inclusion, use the standard
Damiani order in \eqref{eq:Damiani-standard-order}, with $J_m$ in place
of the notation $I_m$ used there.  By Theorem~\ref{thm:any-Damiani},
its ordered monomials form a basis of $U$.

For a finite non-decreasing sequence
$\lambda=(\lambda_1,\ldots,\lambda_s)$, put
$A_\lambda=A_{\lambda_1}\star\cdots\star A_{\lambda_s}$ and define
\begin{equation*}
 \operatorname{aprof}(\lambda)
 =
 (s;\lambda_s,\lambda_{s-1},\ldots,\lambda_1),
 \qquad
 \operatorname{aprof}(\varnothing)=(0).
\end{equation*}
Order these profiles lexicographically.  Write
\begin{equation}
\label{eq:J1-A-block-expansion}
        z=\sum_\lambda A_\lambda f_\lambda,
\end{equation}
where every $f_\lambda$ is a linear combination of ordered monomials in
the $J_m$ and $B_i$.

Suppose that a non-empty $A$-block occurs, and choose $\lambda$ with
maximal $A$-profile among those for which $f_\lambda\neq 0$.  Put
$p=\max\lambda$, and let $e$ be the multiplicity of $p$.  In the
Leibniz expansion of $[J_1,A_\lambda]$, the largest profile is obtained
by replacing one occurrence of $A_p$ by $A_{p+1}$.  If the chosen
occurrence has $j$ further copies of $A_p$ to its right, then
\eqref{eq:adjacent-Damiani-real} contributes the coefficient $q^{-2j}$.
Thus the coefficient of
\begin{equation*}
 A_{\lambda_1}\star\cdots\star A_{\lambda_{s-1}}
 \star A_{p+1}
\end{equation*}
is
\begin{equation}
\label{eq:J1-A-leading-coefficient}
        [2]_q\sum_{j=0}^{e-1}q^{-2j},
\end{equation}
which is non-zero, since
\begin{equation*}
        \sum_{j=0}^{e-1}q^{-2j}
        =
        \frac{1-q^{-2e}}{1-q^{-2}}
        \neq 0.
\end{equation*}  Replacing a smaller index gives a strictly smaller
profile.  The map that raises the largest index by one is injective and
strictly order-preserving.  Finally, $[J_1,f_\lambda]$ contains no
$A$-factor, so it leaves the $A$-profile unchanged.  The maximal-profile
component of $[J_1,z]$ is therefore non-zero, a contradiction.  Hence no
$A$-factor occurs.

We may now write
\begin{equation*}
        z=\sum_\mu g_\mu B_\mu,
\end{equation*}
where $g_\mu\in\F[J_1,J_2,\ldots]$ and
$B_\mu=B_{\mu_1}\star\cdots\star B_{\mu_s}$ for a finite
non-increasing sequence $\mu=(\mu_1,\ldots,\mu_s)$.  Define
$\operatorname{bprof}(\mu)=(s;\mu_1,\ldots,\mu_s)$ and order these
profiles lexicographically.  If a non-empty $B$-block of maximal profile
occurs, let $p$ be its largest index and let $e$ be the multiplicity of
$p$.  Equations \eqref{eq:J1-root-shifts} and
\eqref{eq:adjacent-Damiani-real} show that raising one of these maximal
indices gives the non-zero coefficient
\begin{equation*}
        -[2]_q\sum_{j=0}^{e-1}q^{-2j}.
\end{equation*}
The map that raises the largest index is again injective and strictly
order-preserving, and all other contributions have smaller $B$-profile.
This contradicts
$[J_1,z]=0$.  Thus no $B$-factor occurs, and
$z\in\F[J_1,J_2,\ldots]$.
\end{proof}

\begin{corollary}
\label{cor:G1-centraliser-shuffle}
One has
\begin{equation}
\label{eq:G1-centraliser-shuffle}
        \Cent_U(G_1)
        =
        \F[G_1,G_2,G_3,\ldots].
\end{equation}
\end{corollary}

\begin{proof}
By \eqref{eq:Damiani-shuffle-delta}, each $J_n$ is a non-zero scalar
multiple of $C_n$.  The identity
\eqref{eq:Catalan-alternating} is triangular in the two families
$\{C_n\}_{n\geqslant 1}$ and
$\{\widetilde G_n\}_{n\geqslant 1}$: for $n\geqslant 1$, the
coefficients of $C_n$ and $\widetilde G_n$ are respectively
$(-1)^n[n]_q$ and $[2n]_q$.  Both are non-zero.  Solving successively for
$\widetilde G_n$ and for $C_n$, and using induction on $n$, therefore
gives
\begin{equation}
\label{eq:C-tG-polynomial-equality}
 \F[J_1,J_2,\ldots]
 =
 \F[C_1,C_2,\ldots]
 =
 \F[\widetilde G_1,\widetilde G_2,\ldots].
\end{equation}
Moreover, $C_1=[2]_qxy$, so
\begin{equation*}
        J_1=(q^{-4}-1)\widetilde G_1.
\end{equation*}
The scalar $q^{-4}-1$ is non-zero.  Proposition
\ref{prop:J1-centraliser-shuffle} and
\eqref{eq:C-tG-polynomial-equality} therefore determine the centraliser
of $\widetilde G_1$.  Applying the word-reversal antiautomorphism gives
\eqref{eq:G1-centraliser-shuffle}.
\end{proof}

\subsection{A filtration by the number of \texorpdfstring{$G$}{G}-factors}

Let $\mathcal B$ be the subalgebra of $U$ generated by the negative
alternating words and the words $G_n$:
\begin{equation}
\label{eq:mixed-negative-G-algebra}
 \mathcal B
 =
 \bigl\langle W_{-n},G_m\mid n\in\N,\ m\geqslant 1\bigr\rangle.
\end{equation}
Use the alternating PBW order $W^-<G<W^+$, with increasing index order
inside the first two blocks.

\begin{lemma}
\label{lem:mixed-negative-G-PBW}
The monomials
\begin{equation}
\label{eq:mixed-negative-G-PBW}
 W_{-\lambda_1}\star\cdots\star W_{-\lambda_r}
 \star
 G_{\mu_1}\star\cdots\star G_{\mu_s},
\end{equation}
where $r,s\in\N$ and
\begin{equation*}
 0\leqslant\lambda_1\leqslant\cdots\leqslant\lambda_r,
 \qquad
 1\leqslant\mu_1\leqslant\cdots\leqslant\mu_s,
\end{equation*}
form an $\F$-basis of $\mathcal B$.  Moreover, the subspaces
\begin{equation}
\label{eq:G-length-filtration}
 F_m^G\mathcal B
 =
 \operatorname{span}_{\F}
 \left\{
 \text{the monomials in \eqref{eq:mixed-negative-G-PBW} with }s\leqslant m
 \right\}
\end{equation}
form an exhaustive algebra filtration of $\mathcal B$, with
$F_{-1}^G\mathcal B=0$.
\end{lemma}

Write
\begin{equation*}
 \gr^G\mathcal B
 =
 \bigoplus_{m\geqslant 0}
 F_m^G\mathcal B/F_{m-1}^G\mathcal B.
\end{equation*}

\begin{proof}
The monomials in \eqref{eq:mixed-negative-G-PBW} are linearly independent
by Theorem~\ref{thm:qshuffle-twelve}.  To prove that they span
$\mathcal B$, straighten every occurrence of $G_i\star W_{-j}$.
Inspection of both cases in
\cite[Lemma~15.1]{Terwilliger2019Alternating} shows that the result is a
linear combination of terms $W_{-a}\star G_b$, where $a,b\in\N$ and
$G_0=1$.  In particular, no positive alternating word is created, and the
number of $G$-factors does not increase.  Repeated straightening proves
both assertions.
\end{proof}

Taking word reversal in \eqref{eq:qshuffle-Wminus} gives, for
$n\geqslant 1$,
\begin{equation}
\label{eq:G-x-crossing}
        G_n\star x
        =
        t\,x\star G_n+(1-t)W_{-n}.
\end{equation}

\begin{lemma}
\label{lem:G-length-leading}
Let $b$ be a PBW monomial in \eqref{eq:mixed-negative-G-PBW} containing
exactly $m$ factors from the family $G$.  Then
\begin{equation}
\label{eq:G-length-leading}
        b\star x
        \equiv
        t^m x\star b
        \pmod{F_{m-1}^G\mathcal B}.
\end{equation}
Consequently, on the degree-$m$ component of
$\gr^G\mathcal B$, the operator $C_\alpha=R-\alpha L$ is
\begin{equation}
\label{eq:C-alpha-on-G-graded}
        \gr_m^G(C_\alpha)
        =
        (t^m-\alpha)L.
\end{equation}
The operator $L$ is injective on every homogeneous component of
$\gr^G\mathcal B$.
\end{lemma}

\begin{proof}
Move the final $x$ in $b\star x$ to the left through the $m$ factors
from the family $G$.  If the first term on the right-hand side of
\eqref{eq:G-x-crossing} is chosen at every crossing, the result is
$t^m x\star b$.  This is the unique term with $m$ factors from the family
$G$.  Choosing the correction term at any crossing removes one such
factor.  The subsequent straightening described in
Lemma~\ref{lem:mixed-negative-G-PBW} cannot increase their number, so all
remaining terms lie in $F_{m-1}^G\mathcal B$.  This proves
\eqref{eq:G-length-leading} and \eqref{eq:C-alpha-on-G-graded}.

Finally, left multiplication by $x=W_0$ adds one $W_0$-factor to the
negative block in \eqref{eq:mixed-negative-G-PBW}.  Distinct PBW
monomials remain distinct, so $L$ is injective on each graded component.
\end{proof}

\begin{proposition}
\label{prop:x-centraliser-in-mixed-algebra}
One has
\begin{equation}
\label{eq:x-centraliser-in-mixed-algebra}
        \Cent_{\mathcal B}(x)=\mathcal N.
\end{equation}
\end{proposition}

\begin{proof}
The inclusion $\mathcal N\subseteq\Cent_{\mathcal B}(x)$ follows from
the commutativity of the negative alternating words.  Conversely, let
$0\neq b\in\mathcal B$ satisfy $[b,x]=0$, and let $m$ be its largest
$G$-filtration degree.  If $m\geqslant 1$, then
\eqref{eq:C-alpha-on-G-graded}, with $\alpha=1$, gives
\begin{equation*}
        0=(t^m-1)L(\overline b)
\end{equation*}
in the degree-$m$ component of $\gr^G\mathcal B$.  This is impossible,
since $t^m-1\neq 0$, the class $\overline b$ is non-zero, and $L$ is
injective.  Hence $m=0$, so $b\in\mathcal N$.
\end{proof}

\subsection{The two-generator centraliser}

\begin{theorem}
\label{thm:two-generator-centraliser}
One has
\begin{equation}
\label{eq:two-generator-centraliser}
        \Cent_U\{W_0,W_{-1}\}
        =
        \mathcal N.
\end{equation}
\end{theorem}

\begin{proof}
Since the negative alternating words commute pairwise,
$\mathcal N\subseteq\Cent_U\{W_0,W_{-1}\}$.

For the reverse inclusion, use the $\N^2$-grading of $U$.  Since
$W_0=x$ and $W_{-1}=W$ are homogeneous, their joint centraliser is a
graded subspace.  It therefore suffices to consider a non-zero homogeneous
element $z\in U$ satisfying
\begin{equation}
\label{eq:two-generator-assumption}
        [z,x]=0,
        \qquad
        [z,W]=0.
\end{equation}
Put $r=d(z)$.  By Lemma~\ref{lem:x-centraliser-deletion},
$r\geqslant 0$.  Let $v=D^rz$.  Proposition
\ref{prop:two-generator-reduction} gives $[v,G_1]=0$, and hence
Corollary~\ref{cor:G1-centraliser-shuffle} gives
\begin{equation*}
        v\in\F[G_1,G_2,\ldots]\subseteq\mathcal B.
\end{equation*}
By Lemma~\ref{lem:x-centraliser-reconstruction},
$z=\gamma_rS_rv$.  Each factor $C_{t^j}$ is a linear combination of
left and right multiplication by $x$, and $\mathcal B$ is a subalgebra
containing $x$.  Therefore $z\in\mathcal B$.  The first equality in
\eqref{eq:two-generator-assumption}, together with
Proposition~\ref{prop:x-centraliser-in-mixed-algebra}, now gives
$z\in\mathcal N$.
\end{proof}

\begin{corollary}
\label{cor:negative-shuffle-self-centralising}
The algebra $\mathcal N$ is self-centralising in $U$:
\begin{equation}
\label{eq:negative-shuffle-self-centralising}
        \Cent_U(\mathcal N)=\mathcal N.
\end{equation}
\end{corollary}

\begin{proof}
Since $\mathcal N$ is commutative and contains $W_0,W_{-1}$,
\begin{equation*}
 \mathcal N
 \subseteq
 \Cent_U(\mathcal N)
 \subseteq
 \Cent_U\{W_0,W_{-1}\}
 =
 \mathcal N
\end{equation*}
by Theorem~\ref{thm:two-generator-centraliser}.
\end{proof}

\begin{corollary}
\label{cor:conj167}
One has
\begin{equation}
\label{eq:negative-Oq-self-centralising}
 \Cent_{\Oq}
 \bigl(
 \F[W_0,W_{-1},W_{-2},\ldots]
 \bigr)
 =
 \F[W_0,W_{-1},W_{-2},\ldots].
\end{equation}
\end{corollary}

\begin{proof}
Put
$H^-=\F[W_0,W_{-1},W_{-2},\ldots]\subseteq\Oq$.
By Theorem~\ref{thm:twelve-PBW}, the monomials in the negative alternating
family are linearly independent.  Since this family is commutative,
$H^-$ is a polynomial algebra.

By Theorem~\ref{thm:all-leading}, the initial forms of the negative
alternating generators are non-zero scalar multiples of the corresponding
negative alternating words.  The compatibility of the PBW basis with the
word-length filtration therefore gives
\begin{equation*}
 \gr H^-
 =
 \F[W_0^{\sh},W_{-1}^{\sh},W_{-2}^{\sh},\ldots]
 =
 \mathcal N
 \subseteq U.
\end{equation*}
By Corollary~\ref{cor:negative-shuffle-self-centralising},
$\Cent_U(\gr H^-)=\gr H^-$.  Lemma~\ref{lem:filtered-centraliser} now
gives $\Cent_{\Oq}(H^-)=H^-$.
\end{proof}

Corollary~\ref{cor:conj167} proves
\cite[Conjecture~16.7]{Terwilliger2022OqACE}.

\section{The imaginary alternating centraliser}
\label{sec:centralisers}

For $n\geqslant 1$, put $I_n=B_{n\delta}$, and retain the scalar
convention $I_0=B_{0\delta}=q^{-2}-1$.  We first compare the imaginary
root vectors with the imaginary alternating generators.

\begin{lemma}
\label{lem:imaginary-coordinate-change}
For every $n\geqslant 1$,
\[
 I_n\in\F[\widetilde G_1,\ldots,\widetilde G_n],
 \qquad
 \widetilde G_n\in\F[I_1,\ldots,I_n].
\]
Consequently,
\begin{equation}
\label{eq:imaginary-algebra-equality}
 \F[I_1,I_2,\ldots]
 =
 \F[\widetilde G_1,\widetilde G_2,\ldots],
\end{equation}
and
\begin{equation}
\label{eq:imaginary-commute}
        [I_m,I_n]=0
        \qquad(m,n\geqslant 1).
\end{equation}
Both sides of \eqref{eq:imaginary-algebra-equality} are polynomial algebras
in the displayed generators.
\end{lemma}

\begin{proof}
We first prove by induction on $n$ that
$I_n\in\F[\widetilde G_1,\ldots,\widetilde G_n]$.  In the recursion
\eqref{eq:root-tG-recursion}, the only occurrence of
$I_n=B_{n\delta}$ is the first term
$[n]_qI_n\widetilde G_0$, whose coefficient
$[n]_q\widetilde G_0$ is non-zero.  Every $B_{j\delta}$ occurring in
the sum has $j<n$, with $B_{0\delta}=I_0$ a scalar.  The induction
hypothesis therefore allows us to solve the recursion for $I_n$.

The elements $\widetilde G_n$ commute pairwise, by the image in $\Oq$
of \eqref{eq:Aq10}.  The first assertion now implies
\eqref{eq:imaginary-commute}.

We next prove by induction on $n$ that
$\widetilde G_n\in\F[I_1,\ldots,I_n]$.  In
\eqref{eq:root-tG-recursion}, the unique occurrence of $\widetilde G_n$
corresponds to $(j,k,\ell)=(0,n-1,0)$ and has the non-zero coefficient
$[2n]_q[2]_q^nI_0$.  Every other alternating generator in the recursion
has smaller index and hence belongs to
$\F[I_1,\ldots,I_{n-1}]$ by induction.  Solving for
$\widetilde G_n$ proves the second assertion and
\eqref{eq:imaginary-algebra-equality}.

Finally, Theorem~\ref{thm:BK-nonroot} shows that the ordered monomials in
the $I_n$ are linearly independent.  Hence
$\F[I_1,I_2,\ldots]$ is a polynomial algebra.  The same conclusion for
the $\widetilde G_n$ follows from
\eqref{eq:imaginary-algebra-equality}, or directly from
Theorem~\ref{thm:twelve-PBW}.
\end{proof}

We now introduce a two-sided chain containing all real Baseilhac--Kolb root
vectors.  For $r\in\Z$, define
\[
 R_r=
 \begin{cases}
 B_{r\delta+\alpha_1}, & r\geqslant 0,\\[1mm]
 B_{(-r-1)\delta+\alpha_0}, & r\leqslant -1.
 \end{cases}
\]
Thus $R_{-1}=W_0$ and $R_0=W_1$.  Put
\begin{equation}
\label{eq:kappa-definition}
 \kappa
 =q^{-1}(q-q^{-1})^2[2]_q
 =\frac{(q-q^{-1})(q^2-q^{-2})}{q}.
\end{equation}

We use the Baseilhac--Kolb automorphism $T_0$ in the fixed-parameter form
of \cite[Section~2, equations~(3)--(6)]{Terwilliger2018Action}.  It fixes
$W_0$ and satisfies
\begin{equation}
\label{eq:T0-W1}
 T_0(W_1)
 =
 W_1+
 \frac{qW_0^2W_1-[2]_qW_0W_1W_0+q^{-1}W_1W_0^2}
 {(q-q^{-1})(q^2-q^{-2})}.
\end{equation}
It also satisfies
\begin{equation}
\label{eq:T0-qcommutator}
 T_0\bigl(qW_1W_0-q^{-1}W_0W_1\bigr)
 =qW_0W_1-q^{-1}W_1W_0;
\end{equation}
see \cite[Lemma~2.2]{Terwilliger2018Action}.  All denominators in
\eqref{eq:T0-W1} are non-zero by
Lemma~\ref{lem:nonvanishing-scalars}.  Let $\vartheta=T_0\circ\sigma$,
where $\sigma$ interchanges $W_0,W_1$.

\begin{lemma}
\label{lem:real-chain-relations}
For every $r\in\Z$,
\begin{align}
\vartheta(I_1)&=I_1,
&
\vartheta(R_r)&=R_{r-1},
\label{eq:real-chain-shift-automorphism}\\
[I_1,R_r]
 &=
 \kappa(R_{r-1}-R_{r+1}),
\label{eq:I1-real-shift}\\
R_rR_{r+1}
 &=
 q^{-2}R_{r+1}R_r-I_1.
\label{eq:adjacent-real-straightening}
\end{align}
\end{lemma}

\begin{proof}
We first prove \eqref{eq:I1-real-shift} directly from the real-root
recursions.  For $r=-1$ and $r=0$, equations
\eqref{eq:BK-real0-first} and \eqref{eq:BK-real1-first} give, respectively,
\begin{align*}
R_{-2}
 =R_0+
 \frac{q[I_1,R_{-1}]}{(q-q^{-1})(q^2-q^{-2})},\qquad
R_1
 =R_{-1}-
 \frac{q[I_1,R_0]}{(q-q^{-1})(q^2-q^{-2})}.
\end{align*}
By \eqref{eq:kappa-definition}, these are precisely the two required
instances of \eqref{eq:I1-real-shift}.

For $r\leqslant -2$, write $n=-r-1\geqslant 1$.  Equation
\eqref{eq:BK-real0}, with index $n+1$, becomes
$R_{r-1}=R_{r+1}+\kappa^{-1}[I_1,R_r]$.  For $r\geqslant 1$,
equation \eqref{eq:BK-real1}, with index $r+1$, becomes
$R_{r+1}=R_{r-1}-\kappa^{-1}[I_1,R_r]$.  Both equations are equivalent
to \eqref{eq:I1-real-shift}.

Since $I_1=q^{-2}W_1W_0-W_0W_1$, equation
\eqref{eq:T0-qcommutator} gives $\vartheta(I_1)=I_1$.  Substituting this
expression for $I_1$ into \eqref{eq:BK-real0-first} shows that its
right-hand side is \eqref{eq:T0-W1}.  Hence
$\vartheta(W_0)=T_0(W_1)=R_{-2}$.  Also
$\vartheta(W_1)=T_0(W_0)=W_0$, and therefore
$\vartheta^{-1}(W_0)=W_1$.

For $r\in\Z$, put $S_r=\vartheta^{-r-1}(W_0)$.  Then
$S_{-1}=R_{-1}$, $S_0=R_0$, and $S_{-2}=R_{-2}$.  Applying powers
of $\vartheta$ to the case $r=-1$ of \eqref{eq:I1-real-shift} gives
\[
        [I_1,S_r]=\kappa(S_{r-1}-S_{r+1})
        \qquad(r\in\Z).
\]
Thus the sequences $\{R_r\}_{r\in\Z}$ and
$\{S_r\}_{r\in\Z}$ satisfy the same recurrence and agree at
$r=-1,0$.  Since $\kappa\neq 0$, they agree everywhere.  Hence
$R_r=\vartheta^{-r-1}(W_0)$, which proves
\eqref{eq:real-chain-shift-automorphism}.

Finally, the definition of $I_1$ gives
$R_{-1}R_0=q^{-2}R_0R_{-1}-I_1$.  Applying
$\vartheta^{-r-1}$ proves \eqref{eq:adjacent-real-straightening} for
every $r\in\Z$.
\end{proof}

Order the root vectors by
\begin{equation}
\label{eq:BK-centraliser-order}
 R_r<R_s\iff r>s,
 \qquad
 R_r<I_m,
 \qquad
 I_m<I_n\iff m<n.
\end{equation}
By Theorem~\ref{thm:BK-nonroot}, this order gives a PBW basis of $\Oq$.
Put $\mathcal I=\F[I_1,I_2,I_3,\ldots]$.  For a finite non-increasing
sequence $\boldsymbol r=(r_1,\ldots,r_d)$ of integers, put
$R_{\boldsymbol r}=R_{r_1}\cdots R_{r_d}$, with
$R_{\varnothing}=1$.  Every element of $\Oq$ has a unique expansion
\begin{equation}
\label{eq:BK-centraliser-expansion}
        x
        =
        \sum_{\boldsymbol r}R_{\boldsymbol r}f_{\boldsymbol r},
        \qquad
        f_{\boldsymbol r}\in\mathcal I,
\end{equation}
with only finitely many non-zero coefficients.

Define the \textbf{real profile} of
$\boldsymbol r=(r_1,\ldots,r_d)$ by
$\operatorname{rprof}(\boldsymbol r)=(d;r_1,\ldots,r_d)$, and put
$\operatorname{rprof}(\varnothing)=(0)$.  We order real profiles
lexicographically.  For a real profile $\alpha$, let
\begin{equation}
\label{eq:real-profile-subspace}
 \mathscr R_{<\alpha}
 =
 \operatorname{span}_{\F}
 \left\{
 R_{\boldsymbol s}f
 \,\middle|\,
 \operatorname{rprof}(\boldsymbol s)<\alpha,
 \ f\in\mathcal I
 \right\}.
\end{equation}
For a non-empty sequence $\boldsymbol r=(r_1,\ldots,r_d)$, define
$T(\boldsymbol r)=(r_1+1,r_2,\ldots,r_d)$.  This sequence is again
non-increasing, and $T$ is injective and strictly order-preserving for
the real-profile order.

\begin{lemma}
\label{lem:leading-real-profile}
Let $\boldsymbol r=(r_1,\ldots,r_d)$ be non-empty.  Put $r=r_1$, and
suppose that $r$ has multiplicity $e$ in $\boldsymbol r$.  Then
\begin{equation}
\label{eq:leading-real-commutator}
\begin{aligned}
[I_1,R_{\boldsymbol r}]
\equiv{}&
-\kappa
\left(\sum_{j=0}^{e-1}q^{-2j}\right)
R_{T(\boldsymbol r)}\pmod{\mathscr R_{<\operatorname{rprof}(T(\boldsymbol r))}}.
\end{aligned}
\end{equation}
Here
$R_{T(\boldsymbol r)}=R_{r+1}R_r^{\,e-1}
R_{r_{e+1}}\cdots R_{r_d}$.
\end{lemma}

\begin{proof}
Expand the commutator by the Leibniz rule and use
\eqref{eq:I1-real-shift}.  The target profile
$\operatorname{rprof}(T(\boldsymbol r))$ can arise only by taking the
raising term $-\kappa R_{r+1}$ from one of the first $e$ factors.

Suppose that the selected occurrence has $j$ copies of $R_r$ to its
left.  Moving $R_{r+1}$ to the left through these factors using
\eqref{eq:adjacent-real-straightening} contributes the leading coefficient
$q^{-2j}$.  Every correction term $-I_1$ in this process has two fewer
real factors and hence smaller real profile.  A raising term from a factor
of index smaller than $r$, or any lowering term, likewise gives a smaller
profile.  Summing over the $e$ maximal factors proves
\eqref{eq:leading-real-commutator}.
\end{proof}

\begin{theorem}
\label{thm:imaginary-centraliser}
One has
\begin{equation}
\label{eq:I1-centraliser}
        \Cent_{\Oq}(I_1)
        =
        \mathcal I
        =
        \F[I_1,I_2,I_3,\ldots].
\end{equation}
\end{theorem}

\begin{proof}
The inclusion $\mathcal I\subseteq\Cent_{\Oq}(I_1)$ follows from
\eqref{eq:imaginary-commute}.  Let $x\in\Cent_{\Oq}(I_1)$, and write its
PBW expansion as in \eqref{eq:BK-centraliser-expansion}.  Suppose that a
real factor occurs.  Choose a non-empty sequence $\boldsymbol r$ of
maximal real profile among those with $f_{\boldsymbol r}\neq 0$, and let
its maximal entry have multiplicity $e$.

Since $f_{\boldsymbol r}\in\mathcal I$, it commutes with $I_1$.
Lemma~\ref{lem:leading-real-profile} and the strict order preservation of
$T$ give
\begin{align*}
[I_1,x]
\equiv{}
-\kappa
\left(\sum_{j=0}^{e-1}q^{-2j}\right)
R_{T(\boldsymbol r)}f_{\boldsymbol r}\pmod{\mathscr R_{<\operatorname{rprof}(T(\boldsymbol r))}}.
\end{align*}
No smaller source profile can produce the same target profile, and the
injectivity of $T$ excludes a different source with that target.

The scalar coefficient is non-zero, since
\[
 \sum_{j=0}^{e-1}q^{-2j}
 =
 \frac{1-q^{-2e}}{1-q^{-2}},
\]
and both numerator and denominator are non-zero.  Moreover,
$\kappa\neq 0$ and
$R_{T(\boldsymbol r)}f_{\boldsymbol r}\neq 0$ by the PBW basis.  Thus
the highest-profile component of $[I_1,x]$ is non-zero, a contradiction.
Hence no real factor occurs, so $x\in\mathcal I$.
\end{proof}

\begin{corollary}
\label{cor:conj168}
One has
\begin{equation}
\label{eq:tG1-centraliser}
 \Cent_{\Oq}(\widetilde G_1)
 =
 \F[\widetilde G_1,\widetilde G_2,\widetilde G_3,\ldots].
\end{equation}
\end{corollary}

\begin{proof}
Equation~\eqref{eq:tG1-Bdelta} gives $\widetilde G_1=-qI_1$, so
$\Cent_{\Oq}(\widetilde G_1)=\Cent_{\Oq}(I_1)$.  The result follows
from Theorem~\ref{thm:imaginary-centraliser} and
\eqref{eq:imaginary-algebra-equality}.
\end{proof}

Corollary~\ref{cor:conj168} proves
\cite[Conjecture~16.8]{Terwilliger2022OqACE}.

\begin{corollary}
\label{cor:four-maximal}
The following four polynomial subalgebras of $\Oq$ are self-centralising,
and hence maximal commutative:
\[
\begin{gathered}
 \F[W_0,W_{-1},W_{-2},\ldots],
 \quad
 \F[W_1,W_2,W_3,\ldots],\quad
 \F[G_1,G_2,G_3,\ldots],
 \quad
 \F[\widetilde G_1,\widetilde G_2,\widetilde G_3,\ldots].
\end{gathered}
\]
Moreover,
\[
 \Cent_{\Oq}(G_1)=\F[G_1,G_2,G_3,\ldots].
\]
\end{corollary}

\begin{proof}
The negative alternating algebra is self-centralising by
Corollary~\ref{cor:conj167}.  Applying $\sigma$ gives the same conclusion
for the positive alternating algebra.

Put
$\widetilde{\mathcal H}=
\F[\widetilde G_1,\widetilde G_2,\widetilde G_3,\ldots]$.  This algebra is
commutative and contains $\widetilde G_1$.  Hence
Corollary~\ref{cor:conj168} gives
\[
 \widetilde{\mathcal H}
 \subseteq
 \Cent_{\Oq}(\widetilde{\mathcal H})
 \subseteq
 \Cent_{\Oq}(\widetilde G_1)
 =
 \widetilde{\mathcal H}.
\]
Thus $\widetilde{\mathcal H}$ is self-centralising.  Applying the
antiautomorphism $\dagger$ gives
$\Cent_{\Oq}(G_1)=\F[G_1,G_2,G_3,\ldots]$ and shows that the
$G$-algebra is self-centralising.  Finally, every self-centralising
commutative subalgebra is maximal commutative.
\end{proof}

\end{document}